\documentclass[journal]{IEEEtran}
\usepackage[dvips]{graphicx}
\graphicspath{{./eps/}}
\DeclareGraphicsExtensions{.eps}
\usepackage{subfigure}

\usepackage{mathrsfs}
\usepackage{times}
\usepackage{bm}
\usepackage{subfigure}
\usepackage{color}
\usepackage{amsfonts}
\usepackage{amssymb}
\usepackage{graphicx}
\usepackage{algorithm}
\usepackage{algorithmic}
\usepackage{color}
\usepackage{multirow}
\usepackage{dsfont}
\usepackage{amsmath}
\usepackage{amsthm}
\usepackage{amsfonts}
\usepackage{amssymb}
\newtheorem{lemma}{Lemma}
\newtheorem{theorem}{Theorem}

\newtheorem{remark}{Remark}

\IEEEoverridecommandlockouts

\usepackage{algorithmic}
\usepackage{array}
\usepackage{threeparttable}
\usepackage{enumerate}
\usepackage{balance}
\begin{document}
	\title{Channel Estimation in Massive MIMO Systems with Orthogonal Delay-Doppler Division Multiplexing
	}
\author
{\IEEEauthorblockN{{Dezhi Wang, Chongwen Huang, Xiaojun Yuan, Sami Muhaidat, Lei Liu, Xiaoming Chen, 
Zhaoyang Zhang, \\Chau Yuen, \IEEEmembership{Fellow,~IEEE}, and M{\'e}rouane  Debbah}, \IEEEmembership{Fellow,~IEEE}}\\

 \vspace{-2em}
\IEEEcompsocitemizethanks{
\IEEEcompsocthanksitem 
D.~Wang is with the QianYuan National Laboratory and the College of Information Science and Electronic Engineering, Zhejiang University, Hangzhou, China. (E-mail: dz\_wang@zju.edu.cn)\protect
\protect
\IEEEcompsocthanksitem C. Huang, L. Liu, X. Chen, and Z. Zhang are with the College of Information Science and Electronic Engineering, Zhejiang University, Hangzhou 310027, China. (E-mails:\{chongwenhuang, lei\_liu, chen\_xiaoming, ning\_ming\}@zju.edu.cn). 
\protect 
\IEEEcompsocthanksitem X. Yuan is with the
National Key Laboratory of Wireless Communications, University of Electronic Science and Technology of China, Chengdu 611731, China (E-mail: xjyuan@uestc.edu.cn).\protect
\IEEEcompsocthanksitem S. Muhaidat is with the Department of Electrical Engineering
and Computer Science, Khalifa University, Abu Dhabi, UAE
(e-mail: sami.muhaidat@ku.ac.ae).\protect
\IEEEcompsocthanksitem C. Yuen is with the School of Electrical and Electronics Engineering, Nanyang Technological University, Singapore. (E-mail:
chau.yuen@ntu.edu.sg) \protect
    \IEEEcompsocthanksitem M. Debbah is with the Department of Electrical Engineering and Computer
Science and the KU 6G Center, Khalifa University, Abu Dhabi 127788, UAE,
and also with CentraleSupelec, University Paris-Saclay, 91192 Gif-sur-Yvette,
France (E-mail: Merouane.Debbah@ku.ac.ae).
	}

}
\maketitle
\begin{abstract}
Orthogonal delay-Doppler division multiplexing~(ODDM) modulation has recently been regarded as a promising technology to provide reliable communications in high-mobility situations. Accurate and low-complexity channel estimation is one of the most critical challenges for massive multiple input multiple output~(MIMO) ODDM systems, mainly due to the extremely large antenna arrays and high-mobility environments. To overcome these challenges, this paper addresses the issue of channel estimation in downlink massive MIMO-ODDM systems and proposes a low-complexity algorithm based on memory approximate message passing~(MAMP) to estimate the channel state information~(CSI). Specifically, we first establish the effective channel model of the massive MIMO-ODDM systems, where the magnitudes of the elements in the equivalent channel vector follow a Bernoulli-Gaussian distribution. Further, as the number of antennas grows, the elements in the equivalent coefficient matrix tend to become completely random. Leveraging these characteristics, we utilize the MAMP method to determine the gains, delays, and Doppler effects of the multi-path channel, while the channel angles are estimated through the discrete Fourier transform method. Finally, numerical results show that the proposed channel estimation algorithm approaches the Bayesian optimal results when the number of antennas tends to infinity and improves the channel estimation accuracy by about 30\% compared with the existing algorithms in terms of the normalized mean square error. 
\end{abstract}

\begin{IEEEkeywords}
Channel estimation, massive multiple input multiple output, orthogonal delay-Doppler division multiplexing, memory approximate message passing, discrete Fourier transform 
\end{IEEEkeywords}

\section{Introduction}
With the rapid advancement in wireless communications technology, sixth-generation~(6G) mobile systems are expected to deliver reliable communications across high-mobility scenarios, including high-speed railways~(HSR), unmanned aerial vehicles~(UAV), and low earth orbit satellites~\cite{chen20235g,masaracchia2021uav,ahn2024sensing}. In the fourth-generation~(4G) and fifth-generation~(5G) mobile networks, orthogonal frequency division multiplexing~(OFDM) is widely utilized due to its efficiency and robustness~\cite{hwang2008ofdm,zhang2020joint}. However, in high-mobility scenarios, the Doppler spread associated with time-variant channels leads to severe inter-carrier interference~(ICI)~\cite{khan2024mobility}. This ICI phenomenon significantly degrades the performance of OFDM systems, posing critical challenges for maintaining reliable communication links~\cite{wang2006performance,nagaraj2024inter}.

 To overcome the above challenges, orthogonal time frequency space~(OTFS) multi-carrier~(MC) modulation has been introduced in~\cite{hadani2017orthogonal}, and it has gained significant attention due to its significant benefits in time-varying channels. During the OTFS modulation, the signal in the delay-Doppler~(DD) domain is transformed and mapped to the time frequency~(TF) domain through the inverse symplectic finite Fourier transform~(ISFFT). Subsequently, it is transmitted by utilizing a rectangular pulse in the TF plane, which ensures orthogonality. Consequently, OTFS can be regarded as a form of precoded OFDM. Nevertheless, in practical scenarios, the ideal TF plane pulse described in theory cannot be obtained~\cite{ronny2018orthogonal}. To overcome this limitation and realize an orthogonal pulse in the DD plane, Lin et al.~\cite{lin2022orthogonal} introduced a new multi-carrier (MC) modulation scheme,i.e., orthogonal delay-Doppler division multiplexing (ODDM) modulation. This approach operates directly in the DD domain. By doing so, ODDM modulation effectively circumvents the drawbacks associated with OTFS, offering a more practical alternative for signal transmission in communication systems~\cite{lin2023multi,wang2020exploring}.
 
 In addition, compared with the 5G-NR systems, ODDM modulation operates directly in the DD domain by designing a series of pulse trains, which can overcome the pronounced ICI for higher modulation formats (e.g., 16-QAM). Besides, the ODDM modulation simplified the signal processing in the DD domain, reducing the complexity of data detection compared to the OTFS modulation. At the same time, the ODDM modulation may introduce some additional computational complexity in the transformation between the DD domain and the TF domain compared with OFDM, the performance gain in terms of BER, spectral efficiency, and system capacity in high-speed scenarios far outweighs this additional complexity. Based on the above, we will introduce the literature review and background work about the channel estimation for massive multiple-input multiple-output~(MIMO)-ODDM systems.

\subsection{Literature Review and Background Work}
\subsubsection{Massive  MIMO}
Massive MIMO has been regarded as a promising technology for 5G and future 6G communication systems due to its significant enhancement of spectrum and power efficiency~\cite{gan2022multiple,wei2022multi,Cui2022channel}. To reach the potential of the massive MIMO systems, obtaining accurate channel state information~(CSI) is essential~\cite{wei2021channel,qin2018time,gan2021ris}. However, with the increasing number of antennas, it is challenging to estimate the channel information accurately.  In recent years, some algorithms have been proposed for channel estimation problems in massive MIMO systems. 
Huang \emph{et al.}~\cite{huang2021mimo} considered the millimeter wave~(mmWave) channel estimation utilizing MIMO radar and utilized deep learning~(DL) method for uplink mmWave multi-user MIMO communications. Balevi~\emph{et al.}~\cite{balevi2020massive} considered the multi-cell interference-limited massive MIMO systems and proposed a DL-based channel estimation method, which was close to the performance of the minimum mean square error~(MMSE) estimator.  Han~\emph{et al.}~\cite{han2020channel} focused on the issue of channel estimation, and they explained the non-stationary channel by establishing a connection between subarrays and scatters in the extremely large-scale massive MIMO~(XL-MIMO) systems. Huang~\emph{et al.}~\cite{huang2018iterative} considered the channel estimation challenge in mmWave MIMO systems and proposed an iterative algorithm that combines least squares estimation with a sparse message passing algorithm. In addition, considering that the near-field regions cannot be ignored in XL-MIMO systems, Cui~\emph{et al.}~\cite{cui2022near} employed the polar-domain sparsity to estimate the near-field channel assuming spherical-wavefront property.

\subsubsection{Channel Estimation for Massive MIMO Systems}
Considering the advantages of the massive MIMO technology, it is natural to integrate it into the ODDM system to improve performance. Similar to the MIMO-OTFS and MIMO-OFDM  systems, accurate channel estimation is challenging in massive MIMO-ODDM systems due to the high-dimensional channel parameters~\cite{lin2021channel,meng2024joint,zhang2024sparse,wen2024MF}.  Existing work mainly focuses on the channel estimation problem in massive MIMO-OFDM and massive MIMO-OTFS systems.   Zhao~\emph{et al.}~\cite{zhao2020aoa} investigated OFDM communication systems on HSR, designing two feasible channel estimators including uplink and downlink for OFDM systems with massive antenna arrays located at the base station~(BS). Shen~\emph{et al.}~\cite{shen2019channel} proposed a 3D-structured orthogonal matching pursuit~(3D-OMP) algorithm for estimating the sparse downlink channel information in massive MIMO-OTFS systems, which can attain accurate CSI while minimizing pilot overhead. 
Shi~\emph{et al.}~\cite{shi2021deterministic} introduced a novel CSI acquisition algorithm specifically designed for downlink massive MIMO-OTFS in the presence of the fractional Doppler along with the design of deterministic pilots. Wei~\emph{et al.}~\cite{wei2022off} addressed the channel spreading issues arising from fractional delay and Doppler shifts, introducing an off-grid channel estimation approach that utilizes a sparse Bayesian learning scheme. Different from the OFDM and OTFS systems, Chi~\emph{et al.}~\cite{chi2024interleave} studied the interleave frequency division multiplexing~(IFDM) systems and proposed a detection algorithm based on the memory approximate message passing~(MAMP) scheme. In addition, there are also some researches focusing on the ODDM modulation, Wang~\emph{et al.}~\cite{wang2022joint} proposed a joint channel estimation and data detection algorithm for OTFS-based LEO satellite communication. In addition, a few works focus on channel estimation or signal detection problems in the ODDM systems. Wang~\emph{et al.}~\cite{wang2024exploring} considered the channel estimation and symbol detection problem in ODDM-based integrated sensing and communication~(ISAC) systems and proposed a low-complexity algorithm to estimate the CSI. Cheng~\emph{et al.}~\cite{cheng2024mimo} considered the signal detection problem in the MIMO-MIMO systems and proposed a spatial-based generative adversarial network to detect the signals, which can achieve precise, interference-resistant, and environmentally robust performance.

\subsection{Motivation and Contributions}
However, the existing research may not fully consider the importance of computational complexity in channel estimation problems, particularly in the context of time-varying channels. In high-mobility scenarios, where the CSI fluctuates rapidly, the ability to quickly and efficiently estimate the CSI becomes essential. Therefore, addressing computational complexity is critical for accurate and timely channel estimation in high-mobility scenarios. To overcome the challenges, the MAMP scheme has been proposed to solve the sparse signal recovery problem with low computational complexity~\cite{liu2022memory,chen2023memory,ge2023low,chi2024interleave}, where the techniques of long memory and orthogonalization are used to achieve the potentially Bayes-optimal~\cite{liu2024capacity}. This paper focuses on addressing the channel estimation problem in massive MIMO-ODDM systems based on the MAMP scheme, the main contributions are summarized as follows:
\begin{itemize}
    \item First, we establish the effective channel model for the massive MIMO-ODDM systems, which demonstrates that the equivalent channel vector exhibits sparse characteristics and the magnitudes of its elements are modeled to comply with a Bernoulli-Gaussian distribution.  Furthermore, we provide insights into the behavior of the equivalent coefficient matrix, revealing that as the number of antennas approaches infinity, its elements can be treated as entirely random.
    \item Then, to solve the high dimensional channel estimation problem, we develop a two-stage algorithm with low computational complexity based on the characteristics of the effective channel model. Specifically, the delay, Doppler, and gains of the multi-path channel are calculated by employing MAMP scheme. Meanwhile, the discrete Fourier transforms~(DFT) scheme is utilized to estimate the angles. The proposed algorithm can achieve the Bayes optimal mean square error~(MSE).
    \item Finally, the effectiveness of the proposed algorithm is validated via the numerical results. The proposed algorithm converges to the Bayesian optimal results when the number of antennas increases towards infinity and exhibits a strong bit-error-ratio~(BER) performance in comparison to OTFS modulation. In addition, the proposed algorithm enhances channel estimation accuracy by about 30\% when compared to the existing algorithms.
\end{itemize} 
\begin{figure*}[htbp]
	{	\centering\includegraphics[width=0.9\textwidth]{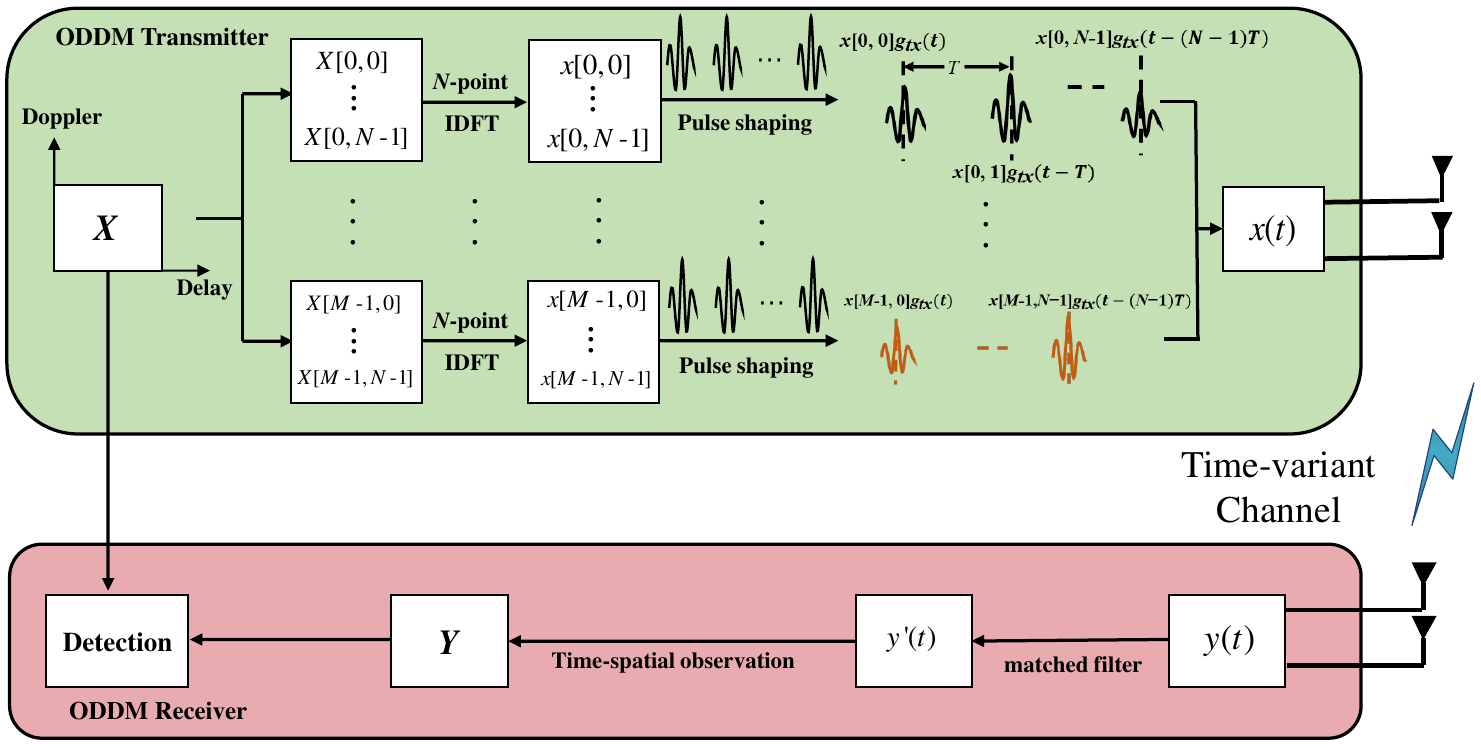}
		\caption{Graphic illustration of ODDM modulation and demodulation}
\label{fig:ODDM_modulation}
	}
\end{figure*}
\vspace{0.5em}
The remainder of the paper is structured as follows: Section II introduces the system model of massive MIMO-ODDM systems, while Section III formulates the channel estimation problem. In Section IV, we present the proposed two-stage low-complexity channel estimation algorithm. Section V offers the numerical results, and finally, Section VI concludes the paper.
\section{System Model}
In this section, we first introduce the physical model of the time-variant channel, and then the ODDM modulation and demodulation process in the SISO systems are presented. Finally, we expand the ODDM modulation in the massive MIMO systems.  As shown in Fig.~\ref{fig:ODDM_modulation}, the ODDM is a DD plane multi-carrier modulation, which is essentially an upsampled staggered multitone~(SMT) modulation. The transmit pulse is orthogonal w.r.t. the DD plane in the ODDM modulation, which can achieve better performance in high-mobility scenarios.

\subsection{Physical Model}
In the time-variant scenarios, define  $W$ and $T$ as the sampling rate and time duration of the received signal, respectively. Thus, the delay and Doppler resolutions are $1/W$ and $1/T$, respectively. The equivalent sampled delay-Doppler domain channel can be written as 
\begin{eqnarray}
h(\tau,\nu)=\sum_{p=1}^Ph_p\delta(\tau-\tau_p)\delta(\nu-\nu_p),
\end{eqnarray}
where $P$ represents the number of paths in the channels, $h_p$, $\tau_p\overset{\triangle}{=}k_p\frac{1}{T}$ and $\nu_p\overset{\triangle}{=}l_p\frac{1}{W}$ are the channel gain, the delay and the Doppler of $p$th path, respectively, where  $k_p,l_p\in\mathbb{Z}$.

\subsection{SISO ODDM Modulation and Demodulation}
In the ODDM modulation, the information is mapped into a $M\times N$ matrix of data symbols within the DD domain, where $M$ represents the number of time slots and  $N$ denotes the number of subcarriers. We define $X[m,n]$ as the data symbols for transmission with $m\in\{0,1,\cdots,M-1\}$ and $n\in\{0,1,\cdots,N-1\}$ indicating the delay and Doppler indices, respectively. Consequently, there are $MN$ data symbols to be modulated. Each MC symbol contains $N$ subcarriers. In the DD domain, the inter-subcarrier spacing is set as $1/NT$. The MC symbols are evenly spaced at intervals of $T/M$, corresponding to the delay resolution. The process of ODDM modulation is shown in~Fig.~\ref{fig:ODDM_modulation}, which is expressed as 
 \begin{eqnarray}\label{ODDM_modulation}
x(t)=\sum_{m=0}^{M-1}\sum_{n=0}^{N-1}X[m,n]g_{tx}\Big(t-m\frac{T}{M}\Big)e^{j2\pi\frac{n}{NT}\Big(t-m\frac{T}{M}\Big)},
 \end{eqnarray}
 where $g_{tx}(t)$ represents the filter at the transmitter, which is a pulse train and expressed as $g_{tx}(t)=\sum_{\hat{n}=0}^{N-1}a(t-\hat{n}T)$. In this expression, $a(t)$ is a real-valued square-root Nyquist pulse with time symmetry. It spans a time duration of $2QT/M$ and adhere the condition $ \int_{-\infty}^{\infty}|a(t)|^2dt=\frac{1}{N}$. Then, we have $\int_{-\infty}^{\infty}|g_{tx}(t)|^2dt=1$. When the integer $2Q\ll M$, it can be proved that $g_{tx}(t)$ is orthogonal w.r.t. the delay and Doppler resolutions, i.e., 
 \begin{eqnarray}\label{orthogonality}
   && A_{g_{tx},g_{tx}}=\int g_{tx}^*(t)g_{tx}\left(t-\frac{T}{M}\right)e^{-j\frac{2\pi n}{NT}\big(t-m\frac{T}{M}\big)}dt \nonumber\\
   &&~~~~~~~~~=\delta(m)\delta(n).
 \end{eqnarray}
 The related proof can be found in~\cite{lin2022orthogonal}.   

 When the transmission signal $x(t)$ passes the time-varying multi-path channel, the received signal at the user from the BS is 
 \begin{eqnarray}
&&y(t)=\sum_{p=1}^Ph_px(t-\tau_p)e^{j2\pi \nu_p(t-\tau_p)}+n_c(t)\nonumber\\
&&=\sum_{p=1}^P\sum_{m=1}^M\sum_{n=1}^Nh_pX[m,n]g_{tx}\left(t-(m+k_p)\frac{T}{M}\right)\nonumber\\
&&~~~\times e^{j2\pi\frac{n+l_p}{NT}\big(t-(m+k_p)\frac{T}{M}\big)e^{j2\pi\frac{l_pm}{MN}}}+n_c(t),
 \end{eqnarray}
 where $n_c(t)\in\mathbb{C}$ is the noise, which follows the complex Gaussian distribution $\mathcal{CN}(0,\delta^2)$, $k_p=\tau_p\frac{M}{T}$ and $l_p=\nu_p NT$ are the delay and the Doppler indexes, respectively, and we assume the channel gain $h_p$ also follows the Gaussian distribution\footnote{In this paper, we consider the channel gain follows a  Gaussian distribution instead of a complex Gaussian distribution. The main reason is that in the case of complex distribution, there are no closed-form solutions for its posterior mean and variance, which increases the difficulty of subsequent calculations.}.

In the ODDM demodulation, the cross-ambiguity function between the received waveform $g_{rx}(t)$ and the received signal $y(t)$ is calculated by using the matched filter as
 \begin{eqnarray}
     A_{g_{rx},y(t)}(\tau,v)=\int e^{-j2\pi v(t-\tau)}g_{rx}^*(t-\tau)y(t)dt,
 \end{eqnarray}

Then, we can derive the discrete form of the received signal
\begin{eqnarray}\label{initial_received_signal}
&&Y[m,n]=A_{g_{rx},y(t)}(\tau,v)|_{\tau=m\frac{T}{M},v=n\frac{1}{NT}}\nonumber\\
&&=\int e^{-j2\pi n\frac{1}{NT}\big(t-m\frac{T}{M}\big)}g_{rx}^*\Big(t-m\frac{T}{M}\Big)y(t)dt\nonumber\\
&&=\left\{
\begin{matrix}
\begin{aligned}
&\sum_ph_pX[m-k_p,\hat{n}]e^{j2\pi\frac{l_p}{MN}(m-k_p)}, \\
&~~~~~~~~~~~~~~~~~~~~~\textrm{if}~m-k_p\geq 0,\\ 
&\sum_ph_pX[M+m-k_p,\hat{n}]e^{j2\pi\left(\frac{l_p}{MN}(m-k_p)-\frac{\hat{n}}{N}\right)}, \\ 
&~~~~~~~~~~~~~~~~~~~~~\textrm{if}~m-k_p<0,
\end{aligned}
\end{matrix}
\right.\\
\end{eqnarray}
where $n'=[n-l_p]_N$. The derivation process can be seen in Appendix~\ref{Section_IO_relation}.

Moreover, we utilize $d_1$ and $d_2$ to denote the delay and Doppler indexes, respectively, with $d_1\in[0,L-1]$ and $d_2\in[-K,K]$, where $L$ and $K$ represent the maximum delay and Doppler indexes, respectively. Therefore, the received signal can be rewritten as 
\begin{eqnarray}
&&Y[m,n]=\sum_{d_1=0}^{L-1}\sum_{d_2=-K}^{K}h_{d_1,d_2}\hat{X}[m'',n'']e^{j2\pi\frac{d_2(m-d_1)}{MN}}\nonumber\\
&&~~~~~~~~~~+N_c[m,n],\nonumber\\
&&=\left\{
\begin{matrix}
\begin{aligned}
&\sum_{d_1=0}^{L-1}\sum_{d_2=-K}^{K}X[m'',n'']H_{d_1,d_2}^{\textrm{DDS}}+N_c[m,n],\\
&~~~~~~~~~\textrm{if}~m''=m-d_1\geq 0, \\ 
&\sum_{d_1=0}^{L-1}\sum_{d_2=-K}^{K}X[m''+M,n'']e^{-j2\pi\frac{n''}{N}}H_{d_1,d_2}^{\textrm{DDS}}\\
&~+N_c[m,n],~~~\textrm{if}~m''=m-d_1< 0,
\end{aligned}
\end{matrix}
\right.
\end{eqnarray}
where $H_{d_1,d_2}^{\textrm{DDS}}=h_{d_1,d_2}e^{j2\pi\frac{d_2}{MN}(m-d_1)}$ is the delay-Doppler spatial~(DDS) channel, $h_{d_1,d_2}$ is the gain corresponding to the path delay $d_1$ and the Doppler shift $d_2$\footnote{For a time-varying channel model, the channel gain $h_p$ at a particular time instant can be related to $h_{d_1,d_2}$ by convolving the impulse response of the channel with a function that models the Doppler shift and delay.}, and $n''=[n-d_2]_N$.
\subsection{Massive MIMO-ODDM Systems}
In massive MIMO-ODDM systems,  the BS employs $N_t$ uniform linear array~(ULA) antennas to serve multiple single-antenna users. For simplicity, we focus on a single user and ignore the influence of channels from other users~\cite{shi2021deterministic}. Consequently, the received signal for the user can be represented as
 \begin{eqnarray}\label{received_siganl}
\bm{y}(t)=\sum_{p=1}^Ph_px(t-\tau_p)e^{j2\pi \nu_p(t-\tau_p)}\bm{a}_{N_t}(\theta_p)+\bm{n}_c(t),
 \end{eqnarray}
 where $\bm{n}_c(t)\in\mathbb{C}^{N_t\times 1}\in\mathcal{CN}(0,\delta^2\bm{I})$ denotes the noise vector, $\bm{a}_{N_t}(\theta_p)=\left[1,e^{j2\pi\alpha_{p}},\cdots,e^{j2\pi(N_t-1)\alpha_{p}}\right]^T, \forall p$ is the transmit array steering vectors, where $\alpha_p=d_{\rm BS}\sin(\theta_p)/\lambda_c$, $\theta_p$, $d_{\rm BS}$, and $\lambda_c$ present physical angle of departure of path $p$, the antenna spacing and the wavelength of the carrier frequency, respectively.  

Similarly, the discrete form of the received signal for the massive MIMO-ODDM systems is expressed as
\begin{eqnarray}\label{IO_relationship}
&&Y[m,n]=\sum_{n_t=0}^{N_t-1}\sum_{d_1=0}^{L-1}\sum_{d_2=-K}^{K}h_{d_1,d_2}\hat{X}_{n_t}[m'',n'']\nonumber\\
&&~~~~~\times e^{j2\pi\Big(\frac{d_2}{MN}(m-d_1)+n_t\alpha_{d_1,d_2}\Big)}+N_c[m,n],\nonumber\\
&&=\left\{
\begin{matrix}
\begin{aligned}
&\sum_{n_t=0}^{N_t-1}\sum_{d_1=0}^{L-1}\sum_{d_2=-K}^{K}X_{n_t}[m'',n'']H_{d_1,d_2,n_t}^{\textrm{DDS}}\\
&~~~+N_c[m,n],~~\textrm{if}~m''=m-d_1\geq 0, \\ 
&\sum_{n_t=0}^{N_t-1}\sum_{d_1=0}^{L-1}\sum_{d_2=-K}^{K}X_{n_t}[m''+M,n'']e^{-j2\pi\frac{n''}{N}}H_{d_1,d_2,n_t}^{\textrm{DDS}},\\
&~~~+N_c[m,n],~~\textrm{if}~m''=m-d_1< 0,
\end{aligned}
\end{matrix}
\right.
\end{eqnarray}
where $X_{n_t}[m,n]$ is the symbol transmitted by  $n_t$th antenna,  $H_{d_1,d_2,n_t}^{\textrm{DDS}}=h_{d_1,d_2}e^{j2\pi\big(\frac{d_2}{MN}(m-d_1)+n_t\alpha_{d_1,d_2}\big)}$  is the corresponding DDS channel\footnote{Similar to the channel gain, $\alpha_{d_1,d_2}$ can be derived from $\alpha_p$ by considering the geometric and physical characteristics of the wireless channel. In this paper,  the delay-Doppler grid point $(d_1,d_2)$ maps to a single dominant path $p$, $a_{d_1,d_2}$ is directly assigned from $\alpha_p$ after quantizing $\tau_p$ and $\nu_p$ to the nearest grid indices $d_1$ and $d_2$.}, and $n''=[n-d_2]_N$. Given the input-output relationship in~\eqref{IO_relationship}, we will concentrate on the issue of channel estimation for massive MIMO-ODDM systems in the following section.

\section{Formulation of the Channel Estimation Problem}
In this section, we develop a formulation for the channel estimation problem in massive MIMO-ODDM systems. We first transform the input-output relationship in~\eqref{IO_relationship} into the vector form and establish the effective channel model for the channel estimation problem.  

As introduced before, the maximum delay and Doppler of the channel are given by $(K-1)\frac{T}{M}$ and $L\frac{1}{NT}$, respectively. We can arrange the $P$ paths of the channel into a matrix~$\bm{G}$ with dimension $(2L+1)\times K$. The gain of the $p$-th path is the non-zero element in matrix $\bm{G}$, i.e., $g(l+L+1,k)$, characterized by a delay of $k\frac{T}{M}$ and Doppler shift of $l\frac{1}{MT}$, respectively. 

Due to the SMT modulation in the ODDM modulation, the $m$-th received ODDM symbol experiences ISI for the $(m-k)$-th symbol due to the path with a delay $k\frac{T}{M}$. In this case, the Doppler shift $l\frac{1}{NT}$ causes a cyclic shift of the interfering $(m-k)$-th symbol initiates at $\frac{(m-k)T}{M}$, and the Doppler effect induces a phase rotation $e^{j2\pi l\frac{m-k}{MN}}$.
	
Building on the preceding analysis, we can express it in the following vector form 
			\begin{eqnarray}\label{vector_form1}
				\bm{y}=\sum_{n_t=0}^{N_t-1}\bm{\Phi}_{n_t}\bm{h}_{n_t}+\bm{n}_c,
			\end{eqnarray}
			where  $\bm{y}=\textrm{vec}(\bm{Y})\in\mathbb{C}^{MN\times 1}$, 
			$\bm{h}_{n_t}=\textrm{vec}(\bm{G}\odot\bm{F}_{n_t} )\in\mathbb{C}^{(2K+1)L\times 1}$, $\big[\bm{F}_{n_t}\big]_{m,n}=e^{j2\pi(n_t-1)\alpha_{m,n}}$,  and $\bm{\Phi}_{n_t}\in\mathbb{C}^{MN\times (2K+1)L}$  is expressed as 
			\begin{align}\label{effective_matrix}
				\bm{\Phi}_{n_t}=\resizebox{2.7in}{!}{$\begin{bmatrix}
						\bm{Z}_{n_t}^0 & \cdots  & \cdots &\bm{Z}_{n_t}^0&\cdots & \cdots  & \bm{Z}_{n_t}^0\\	
						\vdots& \vdots & \vdots & \vdots&  \vdots &  \vdots &\vdots \\
						\bm{Z}_{n_t}^{M-L}& \cdots  & \cdots&\bm{Z}_{n_t}^{M-L}&\cdots & \cdots  & \bm{Z}_{n_t}^{M-L}\\
						\bm{Z}_{n_t}^{M-L+1}&  \cdots & \cdots & \cdots&\cdots& \bm{Z}_{n_t}^{M-L+1} & \hat{\bm{Z}}_{n_t}^{M-L+1} \\
						\vdots& \vdots & \vdots & \vdots&  \vdots &  \vdots &\vdots\\ 
						\bm{Z}_{n_t}^{M-2}& \bm{Z}_{n_t}^{M-2} & \hat{\bm{Z}}_{n_t}^{M-2}  & \cdots&\cdots& \cdots  & \hat{\bm{Z}}_{n_t}^{M-2}\\
						\bm{Z}_{n_t}^{M-1}& \hat{\bm{Z}}_{n_t}^{M-1} & \hat{\bm{Z}}_{n_t}^{M-1}&  \cdots&\cdots& \cdots  & \hat{\bm{Z}}_{n_t}^{M-1}
					\end{bmatrix}$},
			\end{align}
			where 
			\begin{eqnarray}
				&&\bm{Z}_{n_t}^i=[\bm{z}_{n_t}^{i,-K},\cdots,\bm{z}_{n_t}^{i,0},\cdots,\bm{z}_{n_t}^{i,K}],\\
				&&\hat{\bm{Z}}_{n_t}^i=[\hat{\bm{z}}_{n_t}^{i,-K},\cdots,\hat{\bm{z}}_{n_t}^{i,0},\cdots,\hat{\bm{z}}_{n_t}^{i,K}],
			\end{eqnarray}
			where
	$\bm{z}_{n_t}^{i,k}=e^{j2\pi\frac{ki}{MN}}\bm{C}^k\bm{x}_{n_t,i}^T$, $\hat{\bm{z}}_{n_t}^{i,k}=e^{j2\pi\frac{k(i-M)}{MN}}\bm{C}^k\bm{D}\bm{x}_{n_t,i}^T$,$\bm{x}_{n_t,i}$ is $i$-th row of symbols transmitted by $n_t$-th antenna, $\bm{C}$ is the $N\times N$ cyclic permutation matrix, i.e,
			\begin{align}
				\bm{C}=\resizebox{1.2in}{!}{$\begin{bmatrix}
						0& \cdots & 0 & 1 \\	
						1 & \vdots & \vdots &  \vdots \\
						\vdots& \cdots  & \ddots& \vdots \\
						0&  \cdots & 1 &  0
					\end{bmatrix}$},
			\end{align}
			and $\bm{D}=\textrm{diag}\Big\{1,e^{-j\frac{2\pi}{N},\cdots,e^{-j\frac{2\pi(N-1)}{N}}}\Big\}$ the phase rotation term, which is utilized for $\bm{x}_{[m-l]_M}$ when $m-l<0$.

		Further, let $\tilde{\bm{\Phi}}=\left[\bm{\Phi}_0,\bm{\Phi}_1,\cdots,\bm{\Phi}_{N_t-1}\right]$ and $\tilde{\bm{h}}=\Big[\bm{h}_0^T,\bm{h}_{1}^T, \cdots, \bm{h}_{N_t-1}^T \Big]^T$. Then, we can transform~\eqref{vector_form1} into the following effective channel model 
			\begin{eqnarray}\label{IO_relationship2}
				\bm{y}=\tilde{\bm{\Phi}}\tilde{\bm{h}}+\bm{n}_c,
			\end{eqnarray}
			where $\bm{y}\in\mathbb{C}^{MN\times 1}$, $\tilde{\bm{\Phi}}\in\mathbb{C}^{MN\times(2K+1)LN_t}$ is the equivalent coefficient matrix, and $\tilde{\bm{h}}\in\mathbb{C}^{(2K+1)LN_t\times 1}$ is the corresponding equivalent channel vector. The CSI
			$\tilde{\bm{h}}$ contains the information of delay, Doppler, gains, and angles of the multi-path channel. Given that the channel vector contains only $P$ non-zero elements in the channel vector, $\tilde{\bm{h}}$ is a sparse vector. Consequently, the channel estimation problem in the massive MIMO-ODDM system is a sparse recovery problem. Next, we will analyze the properties of the equivalent coefficient matrix $\tilde{\bm{\Phi}}$ and the equivalent channel vector $\tilde{\bm{h}}$. 

Further, let $\tilde{\bm{\Phi}}=\left[\bm{\Phi}_0,\bm{\Phi}_1,\cdots,\bm{\Phi}_{N_t-1}\right]$ and $\tilde{\bm{h}}=\Big[\bm{h}_0^T,\bm{h}_{1}^T, \cdots, \bm{h}_{N_t-1}^T \Big]^T$. Then, we can transform~\eqref{vector_form1} into the following effective channel model 
\begin{eqnarray}\label{IO_relationship2}
\bm{y}=\tilde{\bm{\Phi}}\tilde{\bm{h}}+\bm{n}_c,
\end{eqnarray}
where $\bm{y}\in\mathbb{C}^{MN\times 1}$, $\tilde{\bm{\Phi}}\in\mathbb{C}^{MN\times(2K+1)LN_t}$ is the equivalent coefficient matrix, and $\tilde{\bm{h}}\in\mathbb{C}^{(2K+1)LN_t\times 1}$ is the corresponding equivalent channel vector. The CSI
$\tilde{\bm{h}}$ contains the information of delay, Doppler, gains, and angles of the multi-path channel. Given that the channel vector contains only $P$ non-zero elements in the channel vector, $\tilde{\bm{h}}$ is a sparse vector. Consequently, the channel estimation problem in the massive MIMO-ODDM system is a sparse recovery problem. Next, we will analyze the properties of the equivalent coefficient matrix $\tilde{\bm{\Phi}}$ and the equivalent channel vector $\tilde{\bm{h}}$. 
\begin{remark}[The Characteristic of Equivalent Coefficient Matrix]
As shown in~\eqref{effective_matrix}, there exists a correlation among columns in $\tilde{\bm{\Phi}}$  due to some repeated elements. Therefore, the elements in the equivalent coefficient matrix are not completely random. Fortunately, when the number of antennas tends to infinity, the number of independent and randomly generated data symbols\footnote{The data symbols are actually pilot symbols, which are fixed but randomly generated within a certain defined set. This random generation within a defined set ensures that the pilot symbols can effectively probe the channel state information.} accumulates since each additional antenna can carry independent data symbols. As a result, when the number of antennas satisfies $N_t\gg K$, the elements approach completely random.
\end{remark}
\begin{remark}[The Characteristics of Equivalent Channel Vector]
The non-zero elements in $\tilde{\bm{h}}$ correspond to the channel gain multiplied by the steering vector elements, with the channel gain following a Bernoulli Gaussian distribution. However, since the steering vector elements vary across antennas, the elements in the channel vector $\tilde{\bm{h}}$ do not adhere to the Bernoulli Gaussian distribution. Fortunately, the magnitudes of the steering vector are unity, leading to the magnitudes of the elements in the channel vector being equal to the channel gain, which is distributed according to a Bernoulli-Gaussian distribution.   
\end{remark}

Moreover, traditional sparse signal recovery algorithms encounter a bottleneck of high computational complexity~\cite{liu2022memory}, making them unsuitable for high-mobility scenarios. In this scenario, the CSI changes rapidly, which requires the fast estimation of CSI. As a result, it is essential to develop low-complexity sparse signal recovery algorithms. To tackle the challenges, we estimate the channel information by proposing a two-stage algorithm based on the MAMP and DFT schemes, which eliminates the need for matrix inversion operations, thereby exhibiting low complexity. The details of the algorithm are introduced in the subsequent section.

\section{Proposed Algorithm for the Channel Estimation Problem}
In this section, we present a two-stage low-complexity algorithm for estimating the CSI in the massive MIMO ODDM systems. Specifically, we begin by estimating the channel gains, delay, and Doppler using the MAMP method, followed by estimating the angles through the DFT approach. 
\subsection{Estimation of the Gains, Delay, and Doppler}\label{subsec. estimation gains, delay}
To estimate the gains, delay, and Doppler of the multi-path channels, we initially reformulate the channel estimation problem in~\eqref{IO_relationship2} as
\begin{eqnarray}
\Xi: &  \bm{y}=\tilde{\bm{\Phi}}\tilde{\bm{h}}+\bm{n}_c,\label{input-output}\\
\Sigma: &\tilde{h}_i\sim \textrm{P}_{\tilde{h}}, \forall{i}.
\end{eqnarray}
Here, the channel information $\tilde{\bm{h}}$ refers to the above triplet information, where the elements in $\tilde{\bm{h}}$ follow the Bernoulli-Gaussian distribution. The objective of channel estimation is to find a minimum mean squared error~(MMSE) estimation for channel information $\tilde{\bm{h}}$, where the corresponding MSE converges to 
\begin{eqnarray}
    \textrm{MMSE}\{\tilde{\bm{h}}|\bm{y},\tilde{\bm{\Phi}},\Xi,\Sigma\}=\frac{1}{(2K+1)LN_t}\mathbb{E}[\|\hat{\bm{h}}-\tilde{\bm{h}}\|^2\big],
\end{eqnarray}
where $\hat{\bm{h}}=\mathbb{E}\big[\tilde{\bm{h}}|\bm{y},\tilde{\bm{\Phi}},\Xi,\Sigma\big]$ signifies the posterior estimation of $\tilde{\bm{h}}$. Then, we estimate the gains, delay, and Doppler based on the MAMP method.
\subsubsection{Introduction about the MAMP}
MAMP contains a long-memory linear estimator~(LLE) and a non-linear estimator~(NLE), which are expressed as
\begin{eqnarray}
\textrm{LLE}:& \bm{\mu}_t=\vartheta_t(\tilde{\bm{h}}_1,\cdots,\tilde{\bm{h}}_t),\\
\textrm{NLE}:&\tilde{\bm{h}}_{t+1}=\gamma_t(\bm{\mu}_t).
\end{eqnarray}
In the MAMP process, the long memory is reflected in LLE and $\vartheta_t$ consists of the previous information, i.e., $\tilde{\bm{h}}_k, k\leq t$, where $t$ represents the $t$-th iteration in the MAMP process. In this aspect, the sequential orthogonalization between the current input and output estimation is not adequate to ensure the asymptotic i.i.d. Gaussianity for MAMP. Consequently, a more rigorous orthogonality is demanded, i.e.,  the approximation error $\vartheta_t$ must exhibit orthogonality to all prior estimation errors.


In the process of LLE and NLE, $\vartheta_t$ and $\gamma_t$  are utilized to process the constraints $\Xi$ and $\Gamma$, respectively. Then, we have the following definition
\begin{eqnarray}
\bm{\mu}_t=\tilde{\bm{h}}+\bm{g}_t^{\vartheta},~~\tilde{\bm{h}}_t=\tilde{\bm{h}}+\bm{g}_t^{\gamma},
\end{eqnarray}
where $\bm{g}_t^{\vartheta}$ and $\bm{g}_t^{\gamma}$ are the estimation errors that have a zero mean, and their covariances are expressed as 
\begin{eqnarray}
&&\nu_{t,t'}^{\vartheta}=\frac{1}{MNN_t}\mathbb{E}\big[(\bm{g}_t^{\vartheta})^H\bm{g}_{t'}^{\vartheta}\big],\label{covariance_LE}\\
&&\nu_{t,t'}^{\gamma}=\frac{1}{MNN_t}\mathbb{E}\big[(\bm{g}_t^{\gamma})^H\bm{g}_{t'}^{\gamma}\big]\label{covariance_NLE}.
\end{eqnarray}
The rigorous orthogonality in MAMP scheme is expressed as
\begin{eqnarray}
&&\lim_{MNN_t\rightarrow\infty} \frac{1}{MNN_t}\tilde{\bm{h}}^H\bm{g}_t^{\vartheta}\overset{\textrm{a.s.}}{=}0,~\forall t'\in[1,t],\label{orthogonal1}\\
&&\lim_{MNN_t\rightarrow\infty}\frac{1}{MNN_t}(\bm{g}_{t'}^{\gamma})^{H}\bm{g}_t^{\vartheta}\overset{\textrm{a.s.}}{=}0,~\forall t'\in[1,t],\label{orthogonal2}\\
&& \lim_{MNN_t\rightarrow\infty}\frac{1}{MNN_t}(\bm{g}_{t'}^{\vartheta})^{H}\bm{g}_{t+1}^{\gamma}\overset{\textrm{a.s.}}{=}0,~\forall t'\in[1,t]\label{orthogonal3}.
\end{eqnarray}
Then, we construct the LLE and NLE for the channel estimation problem in the massive MIMO-ODDM systems.
\subsubsection{Design of LLE and NLE}
Based on the LMMSE criterion, we can derive the posterior estimate of $\tilde{\bm{h}}$ as
\begin{eqnarray}\label{LMMSE}
\bm{\mu}_{t}=\tilde{\bm{h}}_t+\zeta_t^{\vartheta}\tilde{\bm{\Phi}}^H(\rho_t\bm{I}+\zeta_t^{\vartheta}\tilde{\bm{\Phi}}\tilde{\bm{\Phi}}^H)^{-1}(\bm{y}-\tilde{\bm{\Phi}}\tilde{\bm{h}}_t),
\end{eqnarray}
where $\rho_t=\frac{\sigma^2}{\nu_{t,t}^{\gamma}
}$, and $\zeta_t^{\vartheta}$ is the orthogonalization parameter. To acquire the channel information, it is necessary to perform matrix inversion, which causes high computational complexity. To significantly reduce the complexity, we first present the following lemma for clarity. 
\begin{lemma}\label{lemma1}
	Given a matrix $\bm{C}$ where $\bm{I}-\bm{C}$ is invertible and the spectral radius $\rho(\bm{C})<1$, then 
				\begin{eqnarray}
					\sum_{i=0}^{\infty}\bm{C}^i=(\bm{I}-\bm{C})^{-1},
				\end{eqnarray}
				Additionally, beginning with $t=1$ and an initial condition $\bm{r}_0=0$,  the sequence $\bm{r}_t=\bm{C}\bm{r}_{t-1}+\bm{x}$ can be utilized to recursively approximate the series sum, which also converges to 
				\begin{eqnarray}
					\lim_{t\rightarrow\infty}\bm{r}_t=(\bm{I}-\bm{C})^{-1}\bm{x}.
			\end{eqnarray}
\end{lemma}
	\begin{proof}
		Based on the SVD method, we have $\bm{C}=\bm{U\Sigma U}^H$, where $\bm{U}$ is unitary matrix. Then we can derive 
				\begin{eqnarray}
					\bm{I}-\bm{C}=\bm{I}-\bm{U\Sigma U}^H=\bm{U}(\bm{I}-\bm{\Sigma})\bm{U}^H.
				\end{eqnarray}
				Since the spectral radius $\rho(\bm{C})<1$, the single values $\sigma_i$ of $\bm{C}$ satisfy $0\leq\sigma_i\leq 1$. Then, the inverse of $\bm{I}-\bm{C}$ can be expressed
				\begin{eqnarray}
					(\bm{I}-\bm{\Sigma})^{-1}=diag\big(\frac{1}{1-\sigma_1},\cdots,\frac{1}{1-\sigma_n}\big). 
				\end{eqnarray}
				Since $\frac{1}{1-x}=\sum_{i=1}^{\infty}x^i~~(|x|<1)$, we have
				\begin{eqnarray}
					(1-\bm{\Sigma})^{-1}=\sum_{i=0}^{\infty}\bm{\Sigma}^i,
				\end{eqnarray}
				which can lead to $(\bm{I}-\bm{C})^{-1}=\sum_{i=0}^{\infty}\bm{U}\bm{\Sigma}^i\bm{U}^H$. 
			
			On the other hand, $\bm{C}^i=(\bm{U\Sigma\bm{U}}^H)^i=\bm{U}\Sigma^i\bm{U}^H$, thus we have  $\sum_{i=0}^{\infty}\bm{C}^i=\sum_{i=0}^{\infty}\bm{U}\bm{\Sigma}^i\bm{U}^H$, which competes the proof.
			
			Consider that $\bm{r}_t=\bm{Cr}_{t-1}+\bm{x}$, we have
				\begin{eqnarray}
					&&\bm{r}_t=\bm{C}^2\bm{r}_{t-2}+\bm{C}\bm{x}+\bm{x}\nonumber\\
					&&~~~~~~~~\cdots\cdots\nonumber\\
					&&~~~=\sum_{i=0}^{t-1}\bm{C}^i\bm{x}\overset{t\rightarrow\infty}{=}(\bm{I}-\bm{C})^{-1}\bm{x}.
			\end{eqnarray}
		\end{proof}
        
Based on \emph{Lemma}~\ref{lemma1}, we can re-express the LLE as 
\begin{eqnarray}
\hat{\bm{\mu}}_t=\big[\bm{I}-\theta_t(\rho_t\bm{I}+\tilde{\bm{\Phi}}\tilde{\bm{\Phi}}^H)\big]\hat{\bm{\mu}}_{t-1}+\kappa_t(\bm{y}-\tilde{\bm{\Phi}}\tilde{\bm{h}}_t),
\end{eqnarray}
where $\kappa_t$ is selected to expedite the convergence of MAMP detector, $\theta_t$ serves as a relaxation parameter, ensuring the spectral radius of $\bm{I}-\theta_t(\rho_t\bm{I}+\tilde{\bm{\Phi}}\tilde{\bm{\Phi}}^H)$ remains below $1$. $\theta_t$ can be determined by
\begin{eqnarray}\label{optimized_theta}
    \theta_t=(\bar{\lambda}+\rho_t)^{-1},
\end{eqnarray}
where $\bar{\lambda}=\frac{\lambda_{\min}+\lambda_{\max}}{2}$, $\lambda_{\min}$ and $\lambda_{\max}$ denote the minimum and maximum eigenvalues of $\tilde{\bm{\Phi}}\tilde{\bm{\Phi}}^H$. When $\lambda_{\min}$ and $\lambda_{\max}$  are difficult to obtain, we can replace them with a low bound and an upper bound, respectively~\cite{liu2022memory}.

For simplicity, we define $\bm{B}$ as $\bar{\lambda}\bm{I}-\tilde{\bm{\Phi}}\tilde{\bm{\Phi}}^H$, then we have $\theta_t\bm{B}=\bm{I}-\theta_t(\rho_t\bm{I}+\tilde{\bm{\Phi}}\tilde{\bm{\Phi}}^H)$. Consequently, we have the radius of $\theta_t\bm{B}$
as 
\begin{eqnarray}
&&\rho(\theta_t\bm{B})=\frac{\lambda_{\max}-\lambda_{\min}}{\lambda_{\min}+\lambda_{\max}+2\kappa_t}\nonumber\\
&&~~~~~~~~<1.
\end{eqnarray}
Hence, the convergence condition is satisfied with the optimized $\theta_t$ in~\eqref{optimized_theta}. Then, we can approximate \eqref{LMMSE} as 
\begin{eqnarray}\label{approximation_LMMSE}
\bm{\mu}_t=\bm{\mathcal{T}}_t\bm{y}+\sum_{k=1}^{t}\bm{R}_{t,k}\tilde{\bm{h}}_k,
\end{eqnarray}
where $\bm{\mathcal{T}}_t =\sum_{k=1}^t\eta_{t,k}\tilde{\bm{\Phi}}^H\bm{B}^{t-k}$, $\bm{R}_{t,k}=-\eta_{t,k}\tilde{\bm{\Phi}}^H\bm{B}^{t-k}\tilde{\bm{\Phi}}$, and 
\begin{eqnarray}
\eta_{t,i} = 
\begin{cases} 
\kappa_{t}, & i = t  \\
\kappa_{i} \prod_{j=i+1}^{t} \theta_{j}, & i < t 
\end{cases}.
\end{eqnarray}
Based on~\eqref{approximation_LMMSE}, the channel estimation requires all the preceding messages, therefore, the current output estimation error must be orthogonal to all previous input estimations to ensure stricter orthogonality requirements.

After obtaining $\theta_t$, we need to optimize $\kappa_t$ to minimize the MSE $\nu_{t,t}^{\vartheta}$ in LLE. Before the optimization, we define $w_t\overset{\triangle}{=}\frac{1}{MNN_t}\text{tr}\{\tilde{\bm{\Phi}}^H\bm{B}^t\tilde{\bm{\Phi}}\}$. According to the definition provided in~\eqref{covariance_LE}, we have 
\begin{eqnarray}
&&\nu_{t,t}^{\vartheta}=\frac{1}{MNN_t}\mathbb{E}\big[(\bm{g}_t^{\vartheta})^H\bm{g}_{t'}^{\vartheta}\big]\nonumber\\
&&\overset{\textrm{a.s.}}{=}\lim_{MNN_t\rightarrow\infty}\frac{1}{MNN_t(\zeta_t^{\vartheta})^2}\bigg(\bm{\mathcal{T}}_t\bm{z}-\sum_{k=1}^{t}\bm{R}_{t,k}\bm{g}_t^{\gamma}\bigg)^H\nonumber\\
&&~~~~\times\bigg(\bm{\mathcal{T}}_{t}\bm{z}-\sum_{k=1}^{t}\bm{R}_{t,k}\bm{g}_{t}^{\gamma}\bigg)\nonumber\\
&&=\frac{1}{MNN_t(\zeta_t^{\vartheta})^2}\bigg[\sigma^2\textrm{tr}\big\{\bm{\mathcal{T}}_t^H\bm{\mathcal{T}}_t\big\}\nonumber\\
&&~~~+\sum_{i=1}^t\sum_{j=1}^t\nu_{i,j}\textrm{tr}\big\{\bm{R}_{t,k}^H\bm{R}_{t,k}\big\}\bigg]\nonumber\\
&&=\frac{1}{(\zeta_t^{\vartheta})^2}\sum_{i=1}^{t}\sum_{j=1}^{t}\eta_{t,i}\eta_{t,j}(\sigma^2 w_{2t-i-j}+\nu_{i,j}^{\gamma}\bar{w}_{t-i,t-j})\nonumber\\
&&=\frac{e_{t,1}\kappa_{t}^{2} - 2e_{t,2}\kappa_{t} + e_{t,3}}{w_{0}^{2}(\kappa_{t} + e_{t,0})^{2}},
\end{eqnarray}
where 
\begin{eqnarray}
&&\bar{w}_{i,j}=\bar{\lambda}w_{i+j}-w_{i+j+1}-w_iw_j,\\ &&\zeta_t^{\vartheta}=\sum_{i=1}^t\eta_{t,i}w_{t-i},\\
&&e_{t,0} = -\sum_{i=1}^{t-1} \frac{-\eta_{t,i}w_{t-i}}{w_{0}},\\
&&e_{t,1} = \sigma^2 w_{0} + \nu_{i,j}^{\gamma}\bar{w}_{0,0}, \\
&&e_{t,2} = -\sum_{i=1}^{t-1} \eta_{t,i}\big(\sigma^2 w_{t-i} + \text{Re}(\nu_{t,i}^{\gamma})\bar{w}_{0,t-i}\big), \\
&&e_{t,3} = \sum_{i=1}^{t-1}\sum_{j=1}^{t-1}\eta_{t,i}\eta_{t,j}(\sigma^2 w_{2t-i-j} + \nu_{i,j}^{\gamma}\bar{w}_{t-i,t-j}).\nonumber\\
\end{eqnarray}
Then, we derive the optimal parameter $\kappa_t$ by minimizing $\nu_{t,t}^{\vartheta}$. Notice that $\nu_{t,t}^{\vartheta}(\kappa_t)=\infty$ when $\kappa_t=e_{t,0}$, therefore, the optimal $\kappa_t^*$ is $\kappa_1^*=1$ and for $t>2$,
\begin{eqnarray}
\kappa_t^* = 
\begin{cases} 
\frac{e_{t,2}e_{t,0}+e_{t,3}}{e_{t,1}e_{t,0}+e_{t,2}}, & \rm{if}~e_{t,1}e_{t,0}+e_{t,2}\neq 0, \\
\infty, & \rm{otherwise} 
\end{cases}.
\end{eqnarray}

On the other hand, the NLE $\gamma_t$ should reduce the Gaussian noise and enhance the quality of the estimate of $\tilde{\bm{h}}$, and the output of the NLE includes the maximum posterior mean and variance. Since the elements in $\tilde{\bm{h}}$ follow a Bernoulli-Gaussian distribution, i.e., 
\begin{eqnarray}
    \tilde{h}_i= b g, \ b \sim {\cal B}(p), \ g \sim {\cal N}(u_g, v_g),
\end{eqnarray}
where the variables $b$ and $g$ are independent. Consequently, the mean and variance of $\tilde{h}_i$ are given by
\begin{eqnarray}
&&\mathbb{E}\{\tilde{h}_i\} = \mathbb{E}\{b\}\mathbb{E}\{g\} = p u_g, \\
&& {\rm Var}\{\tilde{h}_i\} = \mathbb{E}\{\tilde{h}_i^2\} - \mathbb{E}\{\tilde{h}_i\}^2 \nonumber \\
&&= p (1-p) u_g^2 + p v_g.
\end{eqnarray}
We can then state the following theorem:
\begin{theorem}[The derivation of the output of the NLE]
The maximum posterior mean $x_{\rm post}$ and variance $v_{\rm post}$ of $\tilde{h}_i$ is derived as 
\begin{eqnarray}
&& x_{\rm post} = p_{\rm post}\hat{u}_g, \\
&& v_{\rm post} = p_{\rm post}(1-p_{\rm post})\hat{u}_g^2 + p_{\rm post}\hat{v}_g,
\end{eqnarray}
where $p_{\rm post} = \frac{p}{p+c(1-p)\exp(d)}, c= \sqrt{1+\frac{v_g}{\delta^2}}, d = \frac{(y-u_g)^2}{2(v_z+v_g)} - \frac{y^2}{2\delta^2}, \hat{v}_g = \big(v_g^{-1} + (\delta^2)^{-1}\big)^{-1}, \hat{u}_g = \hat{v}_g\big(v_g^{-1}u_g + (\delta^2)^{-1}y\big).$
\end{theorem}
\begin{proof}
    Please refer to Appendix B.
\end{proof}
After obtaining the maximum posterior mean and variance of $\tilde{h}_i$, we need to guarantee the orthogonality between the input error and the output error. 

Finally, to ensure the convergence and enhance the efficiency of the MAMP detector, we utilize  the damping vector $\bm{\varsigma}_{t+1}\overset{\triangle}{=}[\varsigma_{t+1,1},\cdots,\varsigma_{t+1,t+1}]^T$ with $\sum_{i=1}^{t+1}\varsigma_{t+1,i}=1$, then the NLE can be rewritten as 
\begin{eqnarray}
&&\gamma_t(\bm{\mu}_1,\cdots,\bm{\mu}_t)=\frac{1}{\varepsilon_t^{\gamma}}(\hat{\gamma}(\bm{\mu}_1,\cdots,\bm{\mu}_t)-[\bm{\mu}_1,\cdots,\bm{\mu}_t]\bm{u}_t)\nonumber\\
&&~~~~~~~~~~~\times\bm{\varsigma}_{t+1}.
\end{eqnarray}

Define $\mathcal{I}_t$ as the index of effective memories, therefore $\tilde{\mathcal{I}}_t$ is the corresponding complementary. To avoid the trial memory adding the damping process, we set $\bm{\varsigma}_{t+1,i}=0,~i\in\tilde{\mathcal{I}}_t$. Define $\mathcal{U}$ as the covariance of $\tilde{\bm{h}}_i=0,~i\in\mathcal{I}_t$.  We can derive the optimal damping vector as
\begin{eqnarray}
 \bm{\varsigma}_{t+1}^{\mathcal{I}}=\left\{\begin{array}{ll}\frac{\left[\mathcal{\bm{U}}_{t+1}^{\mathcal{I}}\right]^{-1}\bm{1}}{\bm{1}^{\mathrm{T}}\left[\mathcal{\bm{U}}_{t+1}^{\mathcal{I}}\right]^{-1} \bm{1}}, & \text { if } \mathcal{\bm{U}}_{t+1}^{\mathcal{I}} \text { is invertible } \\ {[0, \cdots, 1,0]^{\mathrm{T}},} & \text { otherwise }\end{array}\right.  .
\end{eqnarray}
Then, we can update $\nu_{t,t}^{\gamma}$ according to
\begin{eqnarray}
    \nu_{t+1,t+1}^{\gamma}=\left\{\begin{array}{ll}\frac{1}{\bm{1}^{\mathrm{T}}\left[\mathcal{\bm{U}}_{t+1}^{\mathcal{I}}\right]^{-1} \bm{1}}, & \text { if } \mathcal{\bm{U}}_{t+1}^{\mathcal{I}} \text { is invertible } \\ {\nu_{t,t}^{\gamma},} & \text { otherwise }\end{array}\right. . 
\end{eqnarray}

\subsection{Estimation of the Angles}\label{subsec. estimation angle}
Given the large number of antennas at the BS, we can utilize the DFT to the received signal for the angle estimation, which can investigate the channel sparsity of the angle domain. 
To solve the angle estimation problem, we first give the following lemma
\begin{lemma}
When the number of antennas $N_t\rightarrow\infty$ and $d_{\rm BS}/{\lambda_c}\leq 1$, we have the nonzero element of $\bm{U}_{N_t}^H\bm{A}_{N_t}$ in each column is 
\begin{eqnarray}
    \lim_{N_t\rightarrow \infty}\Big[\bm{U}_{N_t}^H\bm{A}_{N_t}\Big]_{n_p,p}\neq 0, \forall p,
\end{eqnarray}
where $\bm{A}_{N_t}=[\bm{a}(\theta_1), \bm{a}(\theta_2), \cdots, \bm{a}(\theta_P)]$, the (n,m)-th entry of the normalized DFT matrix $\bm{U}_{N_t}$ is $[\bm{U}_{N_t}]_{n,m}=\frac{1}{\sqrt{N_t}}e^{j\frac{2\pi}{N_t}(n-1)(m-1)}$, and 
\begin{eqnarray}
n_p=\left\{\begin{array}{ll} N_t\theta_p+1, & \text { if } \theta_p\in[0,d_{BS}/{\lambda_c})\\ N_t+N_t\theta_p+1 &\text { if } \theta_p\in[-d_{BS}/\lambda,0)\end{array}\right. . 
\end{eqnarray}
\end{lemma}
\begin{proof}
The proof can be seen in~\cite{zhou2022channel}
\end{proof}
 Based on lemma~2, the condition $Nd_{BS}/{\lambda_c}+1\leq N-Nd_{BS}/{\lambda_c}+1$ must hold to avoid the angle ambiguity, then we have $d_{BS}\leq \lambda_c/2$. Applying DFT method to the received signal in~\eqref{received_siganl}, we have 
\begin{eqnarray}
&&\bm{y}_{DFT}=\bm{U}_{N_t}^H\bm{y}(t)\nonumber\\
&&=\sum_{p=1}^Ph_px(t-\tau_p)e^{j2\pi \nu_p(t-\tau_p)}\bm{U}_{N_t}^H\bm{a}(\theta_p)+\bm{U}_{N_t}^H\bm{n}_c(t),\nonumber\\
\end{eqnarray}
Therefore, the received signal $\bm{y}(t)$ after the DFT operation is an asymptotic sparse vector with $P$ nonzero elements. Thus, we can estimate the $P$ angles immediately from the nonzero elements of $\bm{y}_{DFT}$. However, since the number of antennas $N_t$ is finite, $N_t\theta_p, \forall p$ is generally not an integer. Consequently, most of the power of $\bm{y}_{DFT}$ tends to be concentrated in either the $(N_t\theta_p+1)$-th or the $(N_t+N_t\theta_p+1)$-th row, while some power will leak into the neighboring rows.  This phenomenon, referred to as the power leakage effect~\cite{fan2017angle}, which will cause inaccurate estimation for angles. 

Due to the resolution of the DFT being restricted to $\frac{1}{N_t}$, there is a difference between the discrete angle and the actual continuous angle. To enhance the precision of the angle estimation, an angle rotation scheme is employed to address the discrepancies in the DFT process, where the angle rotation matrix  is 
\begin{eqnarray}
    \bm{F}^{\triangle \theta_p}_{N_t}=\textrm{Diag}\big\{1,e^{j\triangle\theta_p},\cdots,e^{j(N_t-1)\triangle\theta_p}\big\},  \forall p,
\end{eqnarray}
where $\triangle\theta_p\in\left[-\frac{\pi}{N_t},\frac{\pi}{N_t}\right],~\forall p$ represent the phase rotation parameters. These parameters are utilized to modify the angle in the transmit steering vector, ensuring minimal power leakage during angle estimation. For example, the $(n,p)$-th element of matrix $\bm{U}_{N_t}^H\bm{F}^{\triangle \theta_p}_{N_t}\bm{A}_{N_t}$ is 
\begin{eqnarray}
    \frac{1}{\sqrt{N_t}}\sum_{k=1}^{N_t}e^{j2\pi(k-1)(\theta_p+\frac{\triangle\theta_p}{2\pi})-\frac{n-1}{N}},
\end{eqnarray}
thus, the optimal phase rotation parameter to avoid power leakages is derived as 
\begin{eqnarray}
    \triangle\theta_p=2\pi\bigg(\frac{n-1}{N}-\theta_p\bigg).
\end{eqnarray}

The ideal phase rotation parameter for $\theta_p$ can be determined by conducting a one-dimensional search to solve the following issue
\begin{eqnarray}\label{rotation_parameter}
\triangle\theta_p=\arg\max_{\triangle\theta}\|\{\bm{U}_{N_t}\}_{:,n_p}^H\bm{A}_{N_t}\bm{y}(t)\|^2_2.
\end{eqnarray}
To better show the impact of the angle rotation operation, we plot the power distribution before and after the angle rotation to better illustrate the sparse properties of $\bm{y}_{DFT}$ in Fig.~\ref{fig:phase_rotation}. Here, the number of antennas is $N_t=512$, and the number of paths between the BS and user is $P=1$ with the angle $\theta=25^{\circ}$. We can see multiple peaks before the phase rotation due to the power leakage. After the phase rotation, more power is located in $\theta=24.95^\circ$, which makes the angle estimation more accurate. In addition, due to the large number of antennas, the difference in angle estimation before and after phase rotation is small. However, when the number of antennas is limited, the estimation errors will be obvious.

\begin{figure}[tbp]
	{	\centering\includegraphics[width=2.7in]{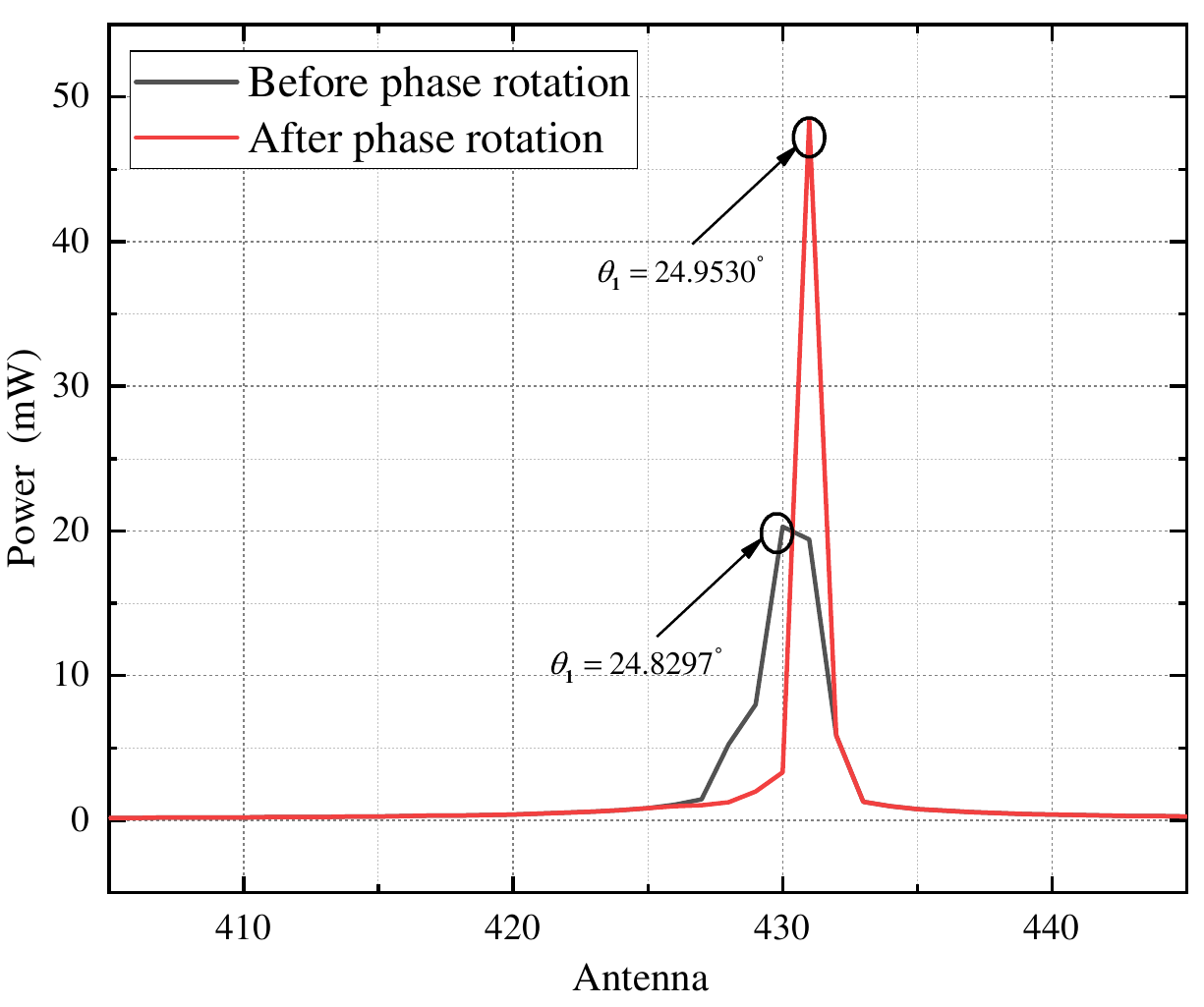}
		\caption{The illustration about the phase rotation of $\bm{y}_{DFT}$}
\label{fig:phase_rotation}
	}
\end{figure}

After obtaining the phase rotation parameters $\triangle\theta_p, \forall p$, we can obtain the angles as follows:
\begin{eqnarray}\label{angle_estimation}
\hat{\theta}_p = 
\begin{cases} 
\arcsin\left(\frac{\lambda_c(n_p-1)}{d_{bs}N_t}-\frac{\lambda_c\triangle\theta_p}{2\pi d_{BS}}\right), & n_p\leq \frac{N_td_{BS}}{\lambda_c}  \\
\arcsin\left(\frac{\lambda_c(n_p-N_t-1)}{d_{bs}N_t}-\frac{\lambda_c\triangle\theta_p}{2\pi d_{BS}}\right), & n_p > \frac{N_td_{BS}}{\lambda_c}  
\end{cases},    
\end{eqnarray}
where 
\begin{eqnarray}
n_p = 
\begin{cases} 
N_t\theta_p+1, & 0\leq\theta_p<d_{\rm BS}/{\lambda_c}  \\
N_t+N_t\theta_p+1, & -d_{\rm BS}/{\lambda_c}\leq \theta_p <0  
\end{cases}.    
\end{eqnarray}
The process of the angle estimation algorithm can be found in Algorithm~\ref{algorithm_AoA estimation}.
\begin{algorithm}[htbp]
	\caption{Angle Estimation Algorithm }\label{algorithm_AoA estimation}
	\begin{algorithmic}[1]
		\STATE Calculate the DFT of the received signal $\bm{y}$: $\bm{y}_{DFT}=\bm{U}_{N_t}^H\bm{y}$
		\STATE Calculate the power of each element in $\bm{y}_{DFT}$;
		\STATE Find the corresponding elements with a power peak
        \STATE Obtain the optimal angle rotation parameters $\{\triangle\theta_p\}_{p=1}^P$ according to \eqref{rotation_parameter}
        \STATE Estimate the angles based on  \eqref{angle_estimation}
		\RETURN $\hat{\theta}_1,\cdots,\hat{\theta}_P$
	\end{algorithmic}
\end{algorithm}

The proposed algorithm contains the MAMP and DFT methods. As shown in Subsec.~\ref{subsec. estimation gains, delay}, the MAMP method can enhance channel estimation performance by leveraging historical iterative information to optimize the convergence and robustness of traditional AMP algorithms. The MAMP scheme can dynamically adjust iteration directions through weighted fusion of current and prior estimates, effectively suppressing error accumulation. By stabilizing convergence toward theoretical bounds, MAMP significantly reduces MSE while maintaining low computational complexity, offering an efficient solution for high-dimensional signal recovery. In addition, the DFT scheme enables efficient angle estimation by transforming the channel into the angular domain, while inherent sparsity highlights dominant path components and suppresses noise, thereby improving the accuracy of the angle estimation. By combining the two schemes, an approximately Bayesian optimal channel estimation result can be obtained.

\subsection{Computational Complexity Analysis}
The computational complexity of the proposed algorithm mainly lies in the MAMP method. In our scenario, the equivalent channel vector $\tilde{\bm{h}}$ is sparse, and the algorithm primarily consists of matrix-vector multiplications. As a result, the proposed MAMP scheme exhibits low complexity.
Suppose that the duration of the iteration time is $T_t$, then the complexity for calculating $\tilde{\bm{\Phi}}\tilde{\bm{\Phi}}^H
\hat{\bm{\mu}}_t$ and $\tilde{\bm{\Phi}}\tilde{\bm{h}}$ is $\mathcal{O}(MN\times(2K+1)LN_tT_t)$, calculating $\sum_{k=1}^{t}\bm{R}_{t,k}\tilde{\bm{h}}_k$ is $\mathcal{O}((MN+(2K+1)LN_t)T_t^2)$, and calculating $\nu_{t,t}^{\vartheta}$ is $\mathcal{O}(T_t^3)$. The computational complexity for calculating $\bm{\varsigma}_{t}$ is negligible compared with the above complexity.  As a result, the overall computational complexity of the proposed algorithm is linear w.r.t. the dimension of the matrix, which is acceptable in practice.

In addition, we will analyze the computation complexity for the baselines, i.e., OMP and 3D-OMP algorithms, in the simulation part. Since the processes of the two algorithms are similar, the computational complexity is close. There, we take the OMP algorithm as an example. The computational complexity of OMP is primarily determined by three steps, including atom selection, least squares solution, and residual update, where the computational complexity of the three step is $\mathcal{O}(AMN(2K+1)LN_t)$, $\mathcal{O}(T^3MN(2K+1)LN_t)$, and $\mathcal{O}(MN(2K+1)LN_t)$, respectively, where $A$ is the number of atoms.
	
Based on the above analysis, the proposed algorithm exhibits a complexity on par with the baseline algorithms, both scaling linearly with the dimension of signals. Notably, the proposed algorithm can achieve near-Bayesian optimality. This enables is to deliver superior channel estimation performance without sacrificing computational efficiency, offering an enhanced balance between complexity and performance. 
\section{Numerical Results}
In this section, we present the numerical results of the proposed channel estimation algorithm. The simulation is conducted with a carrier frequency of 5 GHz and a subcarrier spacing of 15 kHz. The number of time slots $M$ is set to 512, while the number of subcarriers $N$ is 32. The modulation scheme utilized is 4-QAM, and the user's speed is set at 250 km/h. In addition, the number of antennas is 128, and we consider the single-user scenario for simplicity\footnote{In this paper, we consider the single-user scenario for simplicity, and thus the interference from other users is not explicitly modeled. For multi-user scenarios, we can extend the single-user scenarios by utilizing techniques such as multiple access to model the interference from other users.}. In the ODDM modulation, $a(t)$ is configured as a square-root raised cosine pulse with $Q=20$. The channel is characterized by employing the Extended Vehicular A~(EVA) model according to the 3GPP standard~\cite{3gpp2016evolved}. At first,  we utilize the normalized mean square error~(NMSE)\footnote{The NMSE can provide a relative measure that allows for easier comparison across across different simulation scenarios with varying channel power levels. Although many channel taps are with small values, $\|\tilde{\bm{h}}\|_2^2$ is not a very small value based on the results of multiple simulations. The choice of NMSE does not directly align with the theoretical Bayes-optimality framework, but it can accurately reflect the performance of the proposed algorithm, as long as the same metric condition is ensured.} to present the estimation performance, defined as 
\begin{eqnarray}
\textrm{NMSE}=\frac{\|\hat{\bm{h}}-\tilde{\bm{h}}\|_2^2}{\|\tilde{\bm{h}}\|_2^2},
\end{eqnarray}
where $\hat{\bm{h}}$ is the estimated channel information. 
\subsection{Performance Comparison}
In this subsection, we compare the proposed algorithm with the following baselines:
\begin{itemize}
\item \emph{Basedline 1: Traditional Orthogonal Matching Pursuit~(OMP)~\cite{lee2016channel}:} In the traditional OMP, the most relevant atoms are iteratively selected and approximate the target signal through orthogonalization. 
\item \emph{Baseline 2: 3D-Structured Orthogonal Matching Pursuit~(3D-OMP)~\cite{shen2019channel}:} Unlike the traditional OMP algorithm, the most relevant atoms are arranged as a tensor in the 3D-OMP algorithms. 
\end{itemize}

In Fig.~\ref{fig:NMSE_performance_vs_SNR},we present the NMSE performance comparison of various algorithms under different signal noise ratios~(SNRs). The NMSE of all algorithms shows an increased trend with higher SNRs, the reason is obvious because the higher SNRs can improve the performance of the channel estimation. Moreover, the proposed algorithm outperforms both traditional OMP and 3D-OMP schemes. This superior performance can be attributed to the fact that the proposed algorithm can achieve near-Bayes-optimal results because the characteristics of the equivalent coefficient matrix, i.e., approaching but not being completely random.

\begin{figure}[tbp]
	{	\centering\includegraphics[width=2.7in]{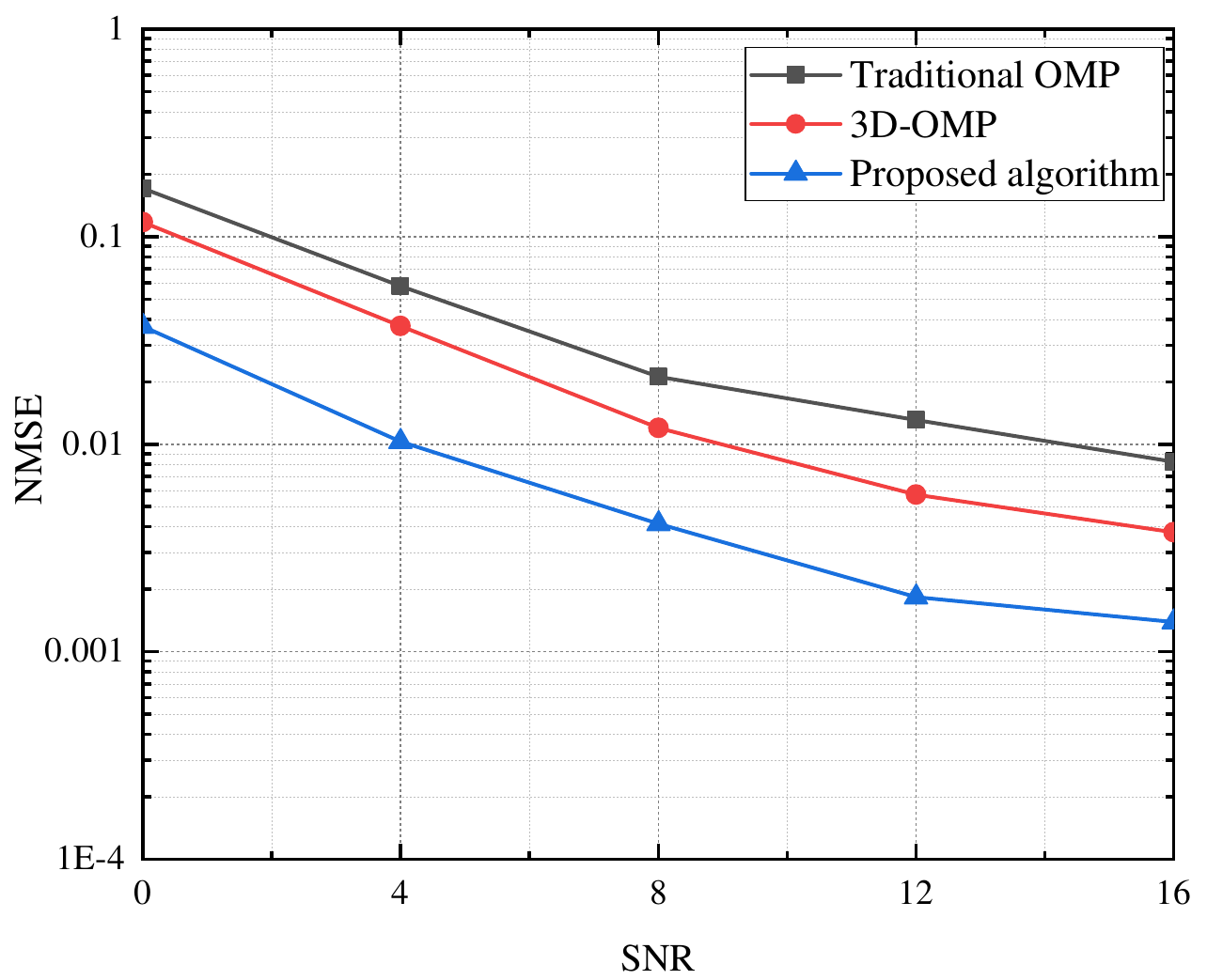}
		\caption{The NMSE performance for different SNRs}
\label{fig:NMSE_performance_vs_SNR}
	}
\end{figure}

In Fig.~\ref{fig:NMSE_vs_speed}, the NMSE performance for different user speeds is illustrated. It is evident that the NMSE performance of the three algorithms remains almost stable across different user speeds, showcasing the robustness of the proposed ODDM modulation. Furthermore, the proposed channel estimation algorithm demonstrates notably superior NMSE performance when compared to the conventional OMP and 3D-OMP algorithms, underlining its effectiveness in diverse user speed scenarios.  This highlights the potential of the proposed algorithm for reliable and robust channel estimation in practical applications. 

\begin{figure}[tbp]
	{	\centering\includegraphics[width=2.7in]{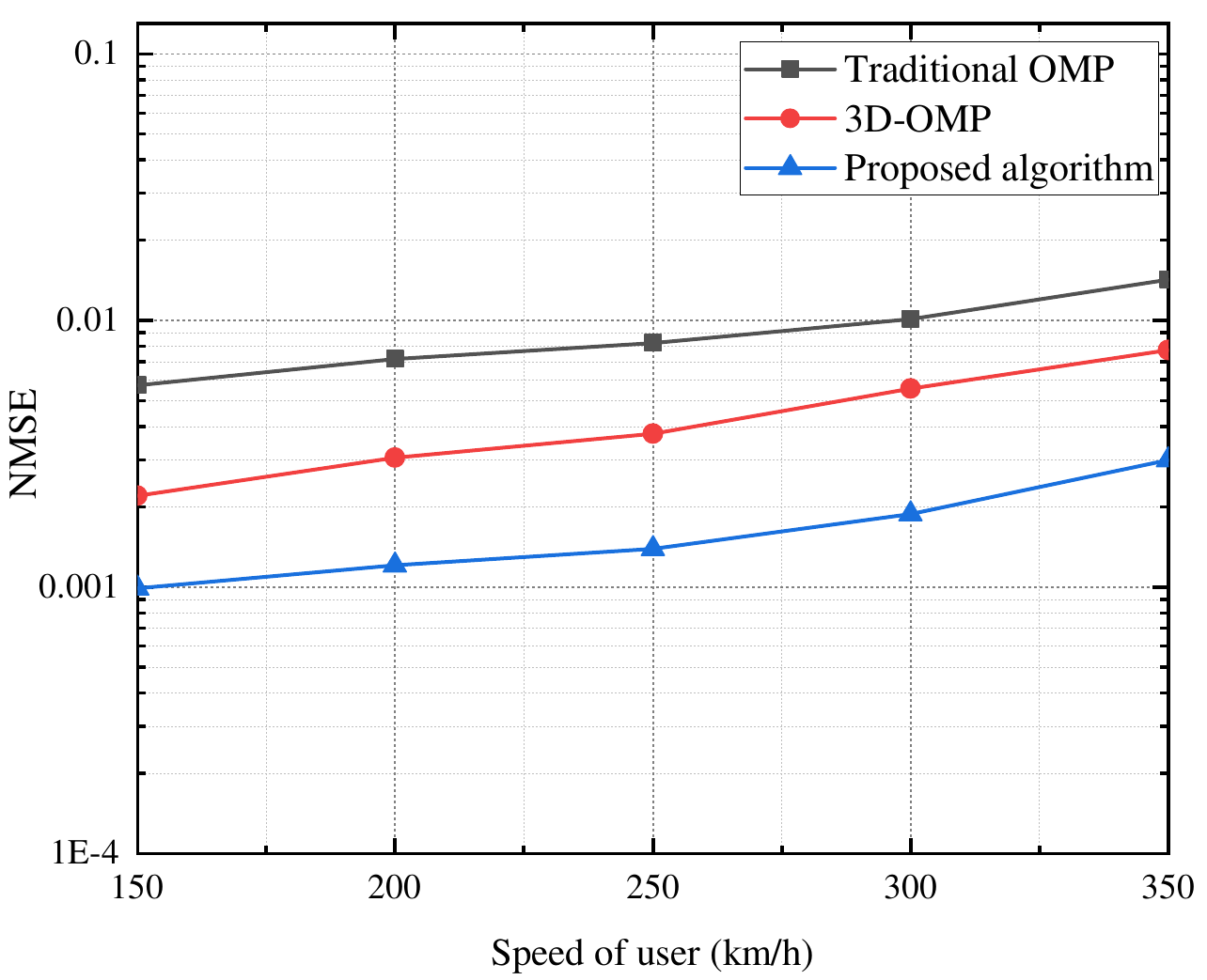}
		\caption{The NMSE performance for different speeds}
\label{fig:NMSE_vs_speed}
	}
\end{figure}

In Fig.~\ref{fig:NMSE_vs_antenna}, we evaluate the NMSE performance of the three algorithms with varying numbers of antennas. For both traditional OMP and 3D-OMP algorithms, the NMSE improves as the number of antennas increases. When the number of antennas exceeds 100, the NMSE performance of traditional OMP converges closely to that of 3D-OMP is close. This convergence can be attributed to the performance degradation in the 3D-OMP algorithm caused by an insufficient number of pilots. In contrast, the proposed algorithm demonstrates remarkable stability, with its NMSE performance showing no significant variation across different numbers of antennas. This consistency underscores the robustness and reliability of the proposed algorithm under different antenna configurations.

In Fig.~\ref{fig:NMSE_vs_paths}, we compare the NMSE performance for different numbers of paths. It can be seen that the NMSE increases with the increased number of paths. This phenomenon can be attributed to the growing estimation error when more paths are involved, which affects the NMSE performance across the three algorithms. Similarly, the proposed algorithm also outperforms the traditional OMP and OMP schemes. Based on the aforementioned simulation results, it can be observed that the proposed algorithm outperforms baseline algorithms in different scenarios. Therefore, it can be concluded that the proposed algorithm possesses better accuracy.

\begin{figure}[tbp]
	{	\centering\includegraphics[width=2.7in]{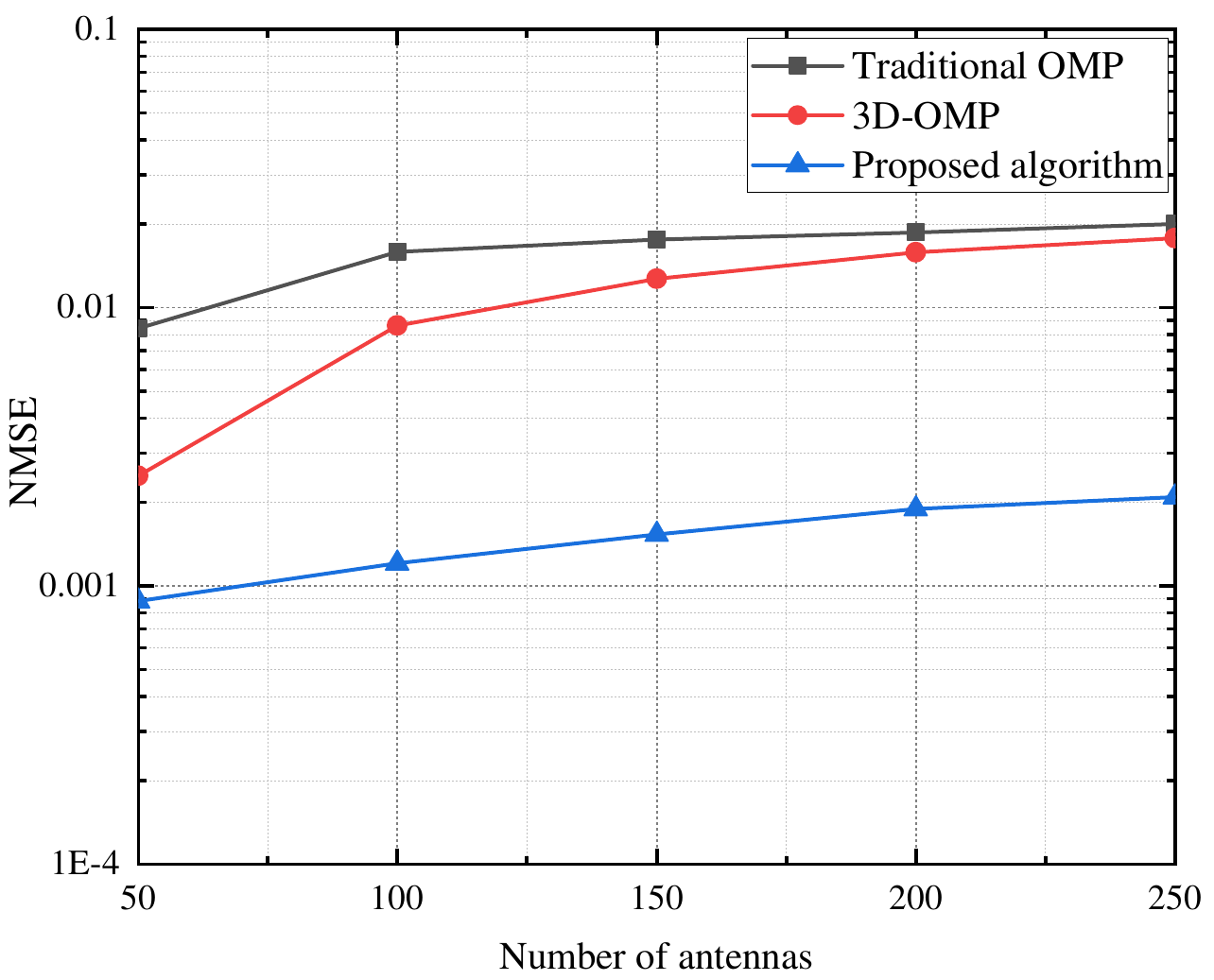}
		\caption{The NMSE performance for different number of antennas}
\label{fig:NMSE_vs_antenna}
	}
\end{figure}

\begin{figure}[tbp]
	{	\centering\includegraphics[width=2.7in]{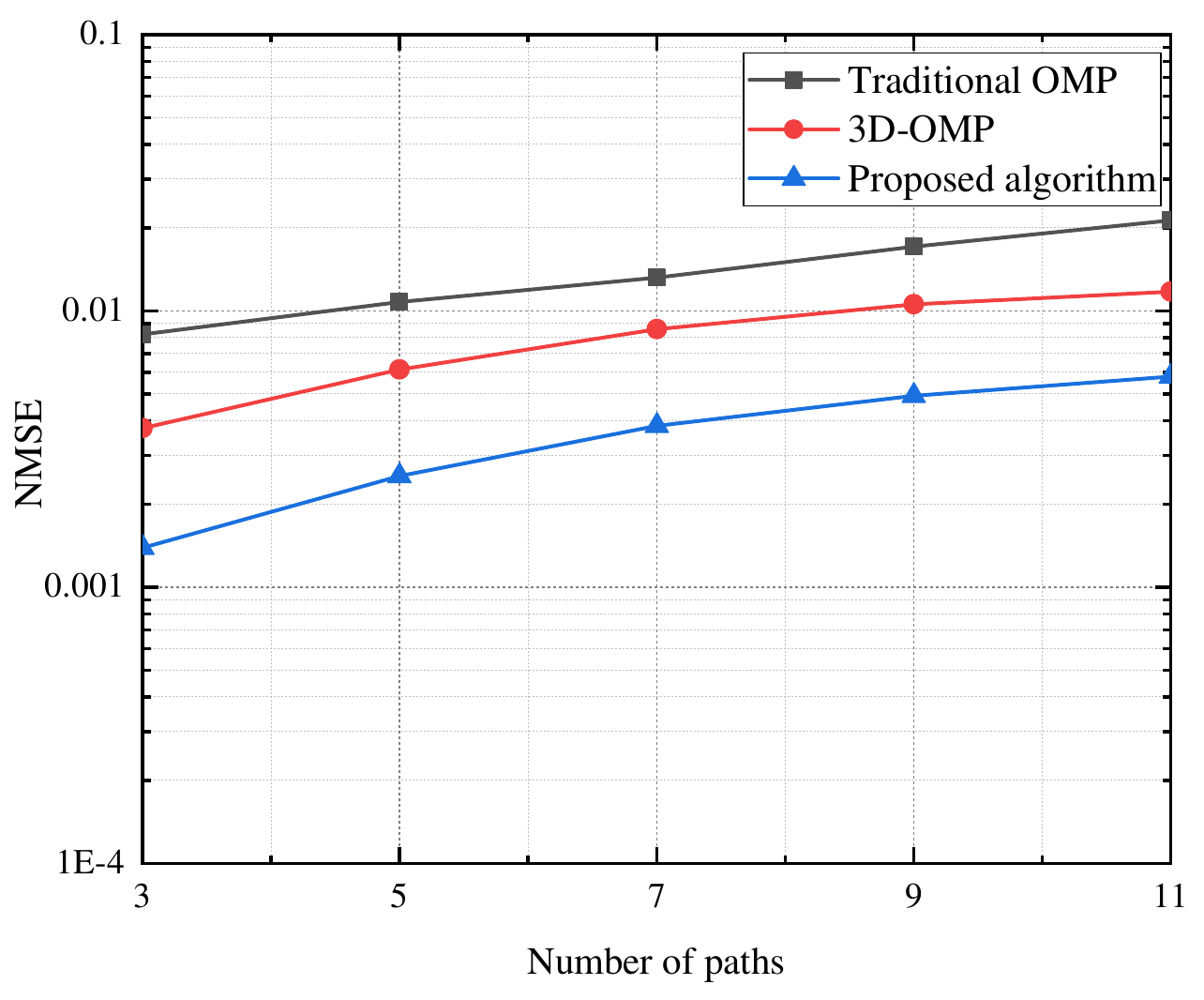}
		\caption{The NMSE performance for different number of paths}
		\label{fig:NMSE_vs_paths}
	}
\end{figure}

\subsection{Performance Analysis}
In this subsection, we delve into the performance evaluation of the proposed algorithm by conducting a comprehensive analysis. This includes an in-depth examination of its convergence behavior and a comparative study against traditional MAMP and OTFS modulation schemes. Through this analysis, we aim to highlight the advantages and potential improvements offered by the proposed algorithm in various operational scenarios.

We initially present the convergence analysis of the proposed algorithm in Fig.~\ref{fig:convergence analysis}. It is evident that the NMSE performance will gradually decrease with the number of iterations, ultimately converging towards a stable value. In scenarios featuring a higher number of antennas, the convergence of the proposed channel estimation algorithm will be slower attributed to the large amount of data. In addition, the proposed algorithm converges to almost the same result for different numbers of antennas.

\begin{figure}[tbp]
	{	\centering\includegraphics[width=2.7in]{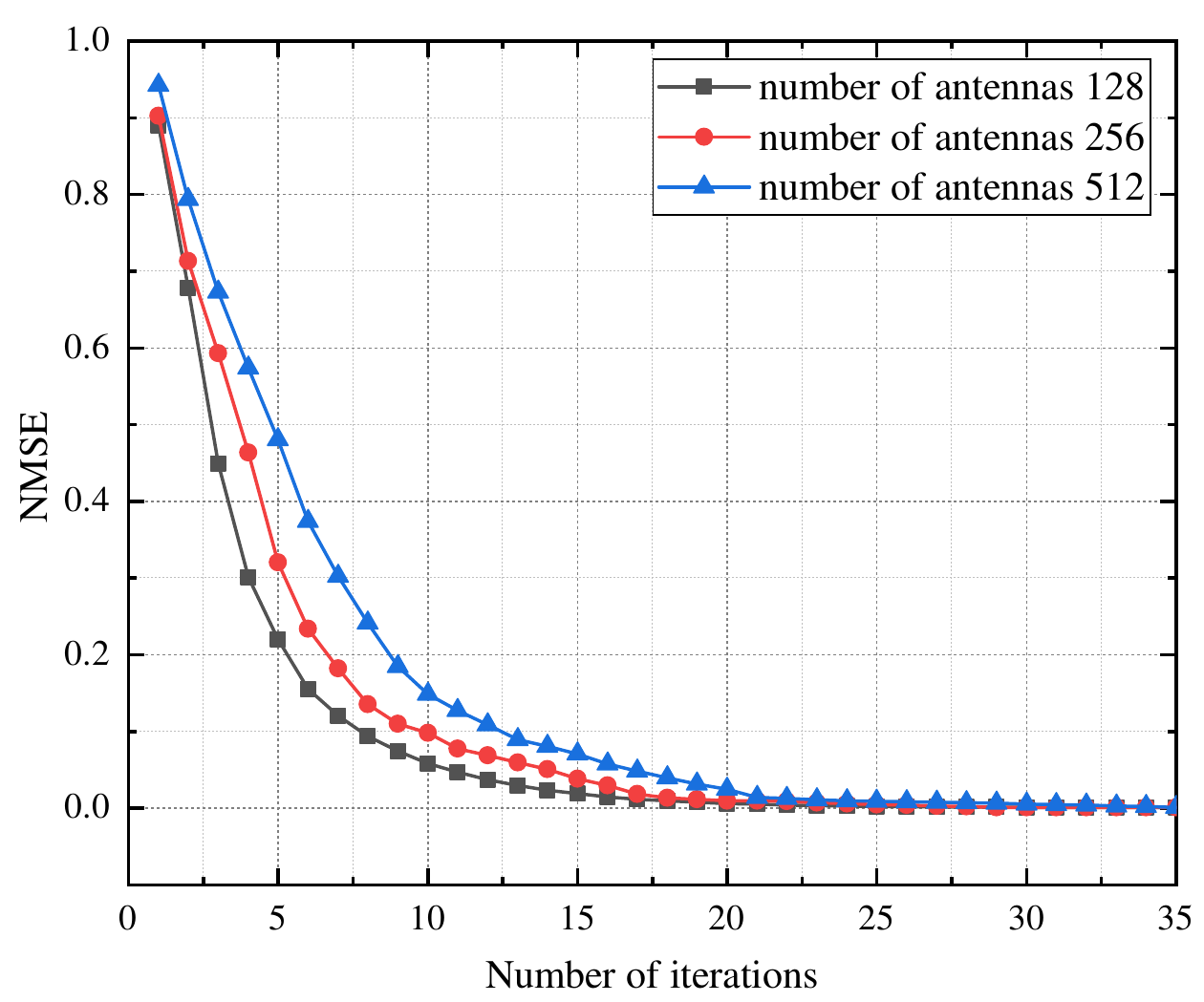}
		\caption{The convergence performance of the proposed algorithm}
\label{fig:convergence analysis}
	}
\end{figure}

\begin{figure}[tbp]
	{	\centering\includegraphics[width=2.7in]{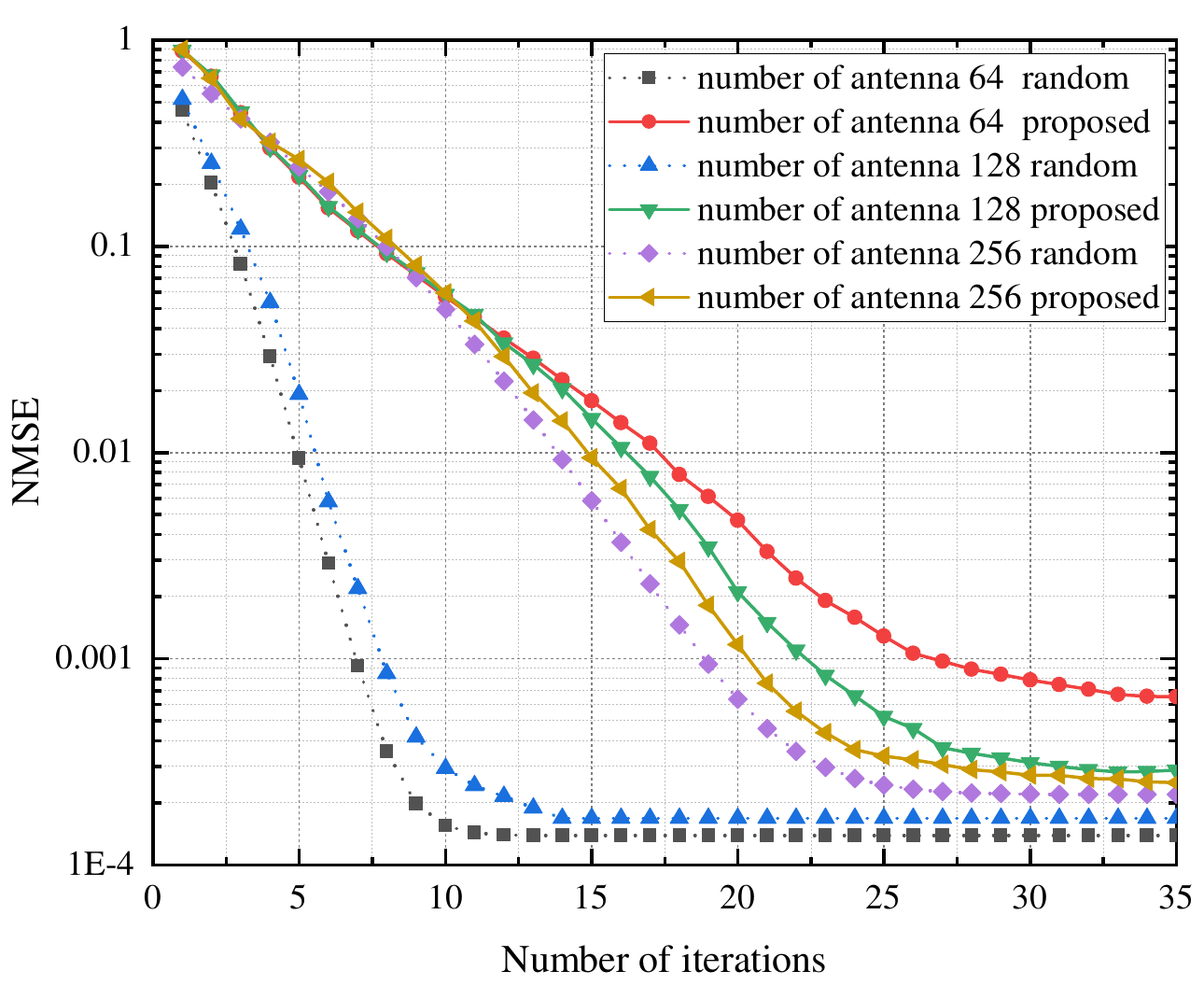}
		\caption{The performance comparison with the traditional MAMP scheme}
\label{fig:performance_comparision_original_MAMP}
	}
\end{figure}

\begin{figure}[tbp]
	{	\centering\includegraphics[width=2.7in]{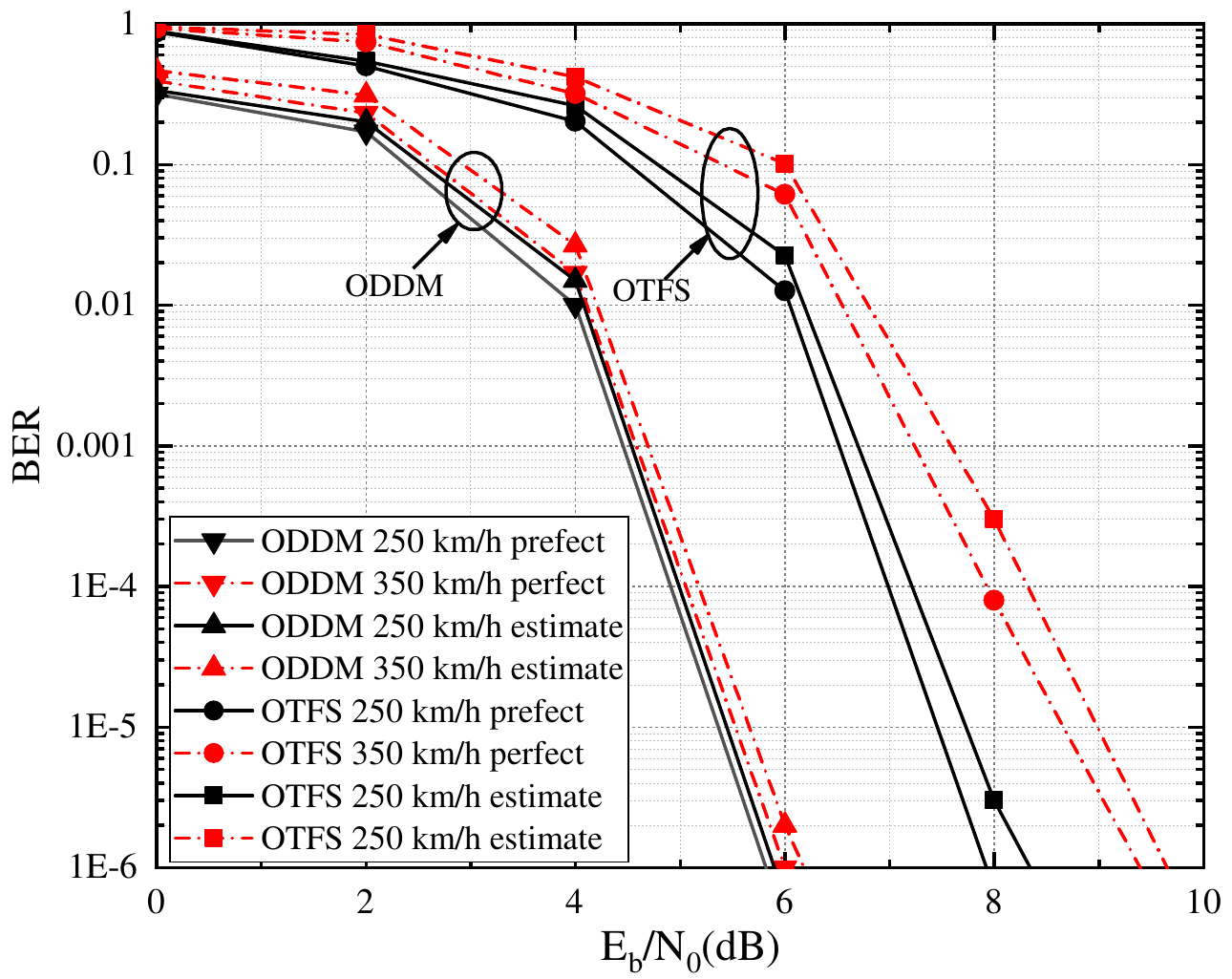}
		\caption{The BER performance versus SNRs for different speeds}
		\label{fig:BER_different_speed}
	}
\end{figure}

\begin{figure}[tbp]
	{	\centering\includegraphics[width=2.7in]{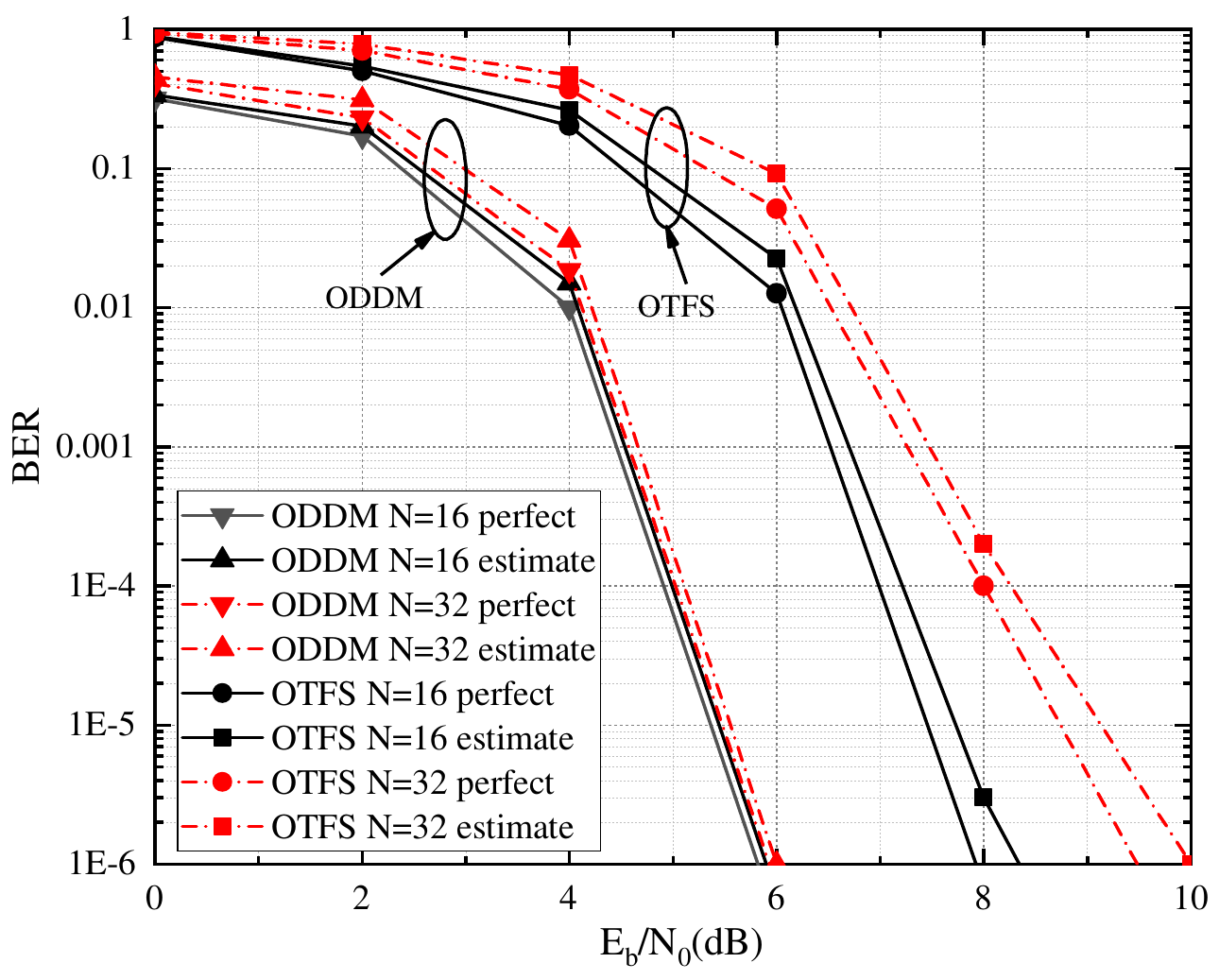}
		\caption{The BER performance versus SNRs for different numbers of subcarriers}
    \label{fig:BER_different_subcarriers}
	}
\end{figure}

In Fig.~\ref{fig:performance_comparision_original_MAMP}, we compare the NMSE performance of the proposed algorithm with the traditional MAMP scheme across different numbers of antennas. In the traditional MAMP scheme, the elements in the effective channel matrix are entirely randomly generated and can achieve Bayesian optimal results. From the results, we can see that the NMSE increases slowly with the increased number of antennas in the traditional MAMP scheme.  It is evident that as the number of antennas increases, the NMSE performance of the proposed algorithm converges towards that of the traditional MAMP scheme. This convergence trend implies that, as the number of antennas approaches infinity, the effective channel matrix can be approximated as being completely random. Consequently, the proposed algorithm demonstrates robustness and efficiency in scenarios with different numbers of antennas, eventually aligning closely with the Bayesian optimal results.

Finally, we present the bit-error-rate~(BER) comparison for different SNRs under different speeds, numbers of subcarriers, different modulations, and different channels. In this scenario, we consider the scenarios that the CSI obtained by the proposed algorithm and perfect CSI, respectively, and the symbol detection is performed by employing the orthogonal approximate message passing~(OAMP) method, which can also achieve Bayes-optimal~\cite{ma2017orthogonal}.  

In Figs.~\ref{fig:BER_different_speed} and~\ref{fig:BER_different_subcarriers}, we consider the BER performance for different speeds and subcarriers. It is evident that the BER rises with the increasing speeds while decreasing subcarriers for both ODDM modulation and OTFS modulation. Compared with the OTFS modulation, the ODDM can achieve better BER performance and robustness for different speeds and subcarriers. As shown in the figures, the BER performance also degrades for perfect CSI. This is because the rapid change of CSI makes it no longer "perfect" during the data transmission. When the speed increases, the Doppler shift increases, causing significant changes in channel characteristics within a short time. Further, the perfect CSI can achieve better BER performance, while the difference between the estimated CSI and the perfect CSI is too small. Combining the results in Figs.~9 and~10, we can conclude that the proposed algorithm achieves good performance and the ODDM modulation is robust for different scenarios.

The BER performance for higher-order modulations is presented in Fig.~\ref{fig:BER_different_modulation}. It can be seen that the higher the modulation order, the greater the corresponding BER, and the worse the signal detection performance. This is because the constellation points of higher-order modulations (such as 16QAM and 64QAM) have a smaller spacing, resulting in a decrease in the anti-noise ability. Moreover, they are more sensitive to channel distortion and synchronization errors, which leads to the signal detection performance being significantly inferior to that of lower-order modulations.
		\begin{figure}[tbp]
		{	\centering\includegraphics[width=2.7in]{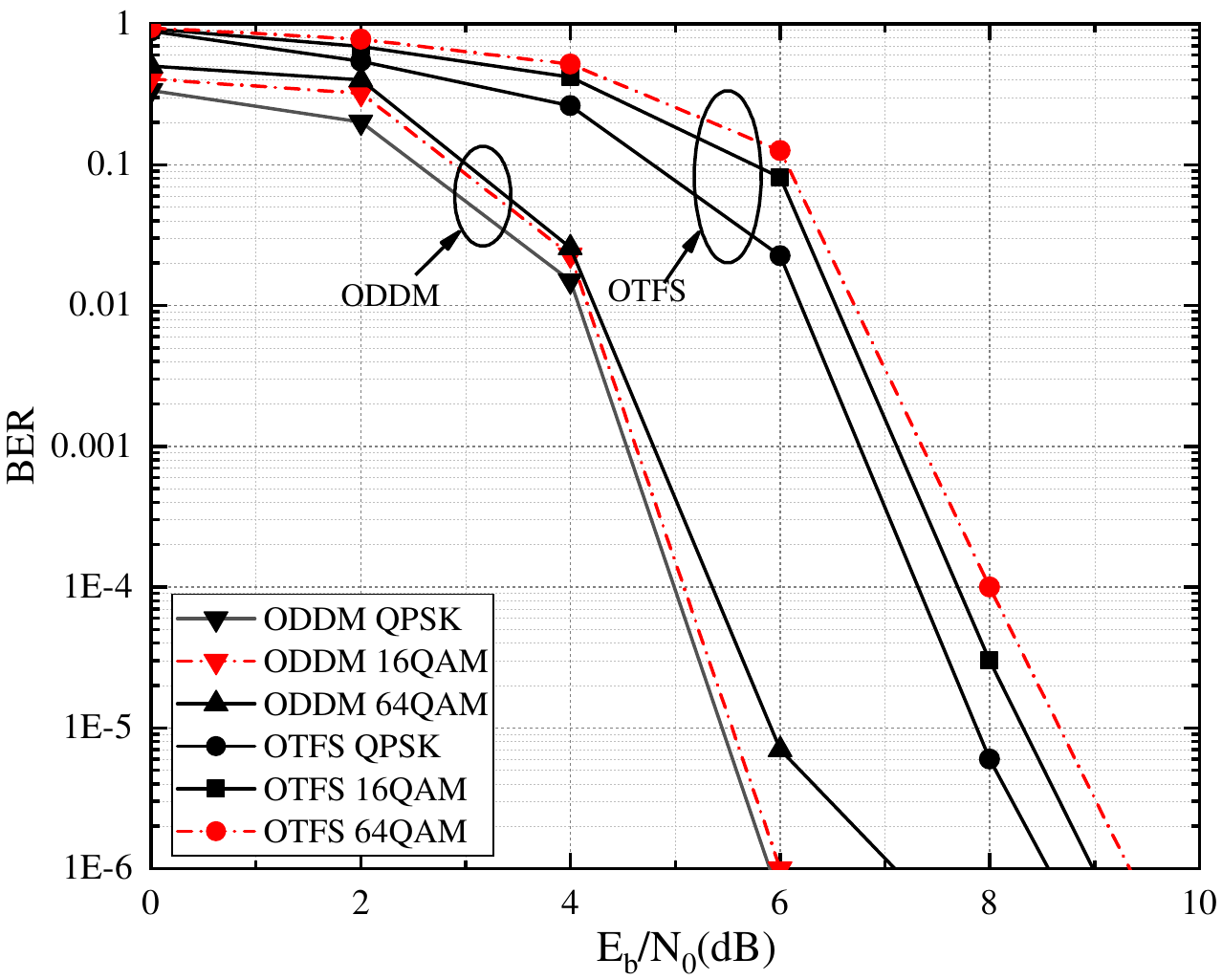}
			\caption{The BER performance versus SNRs for different modulations}
			\label{fig:BER_different_modulation}
		}
	\end{figure}

\begin{figure}[tbp]
	{	\centering\includegraphics[width=2.7in]{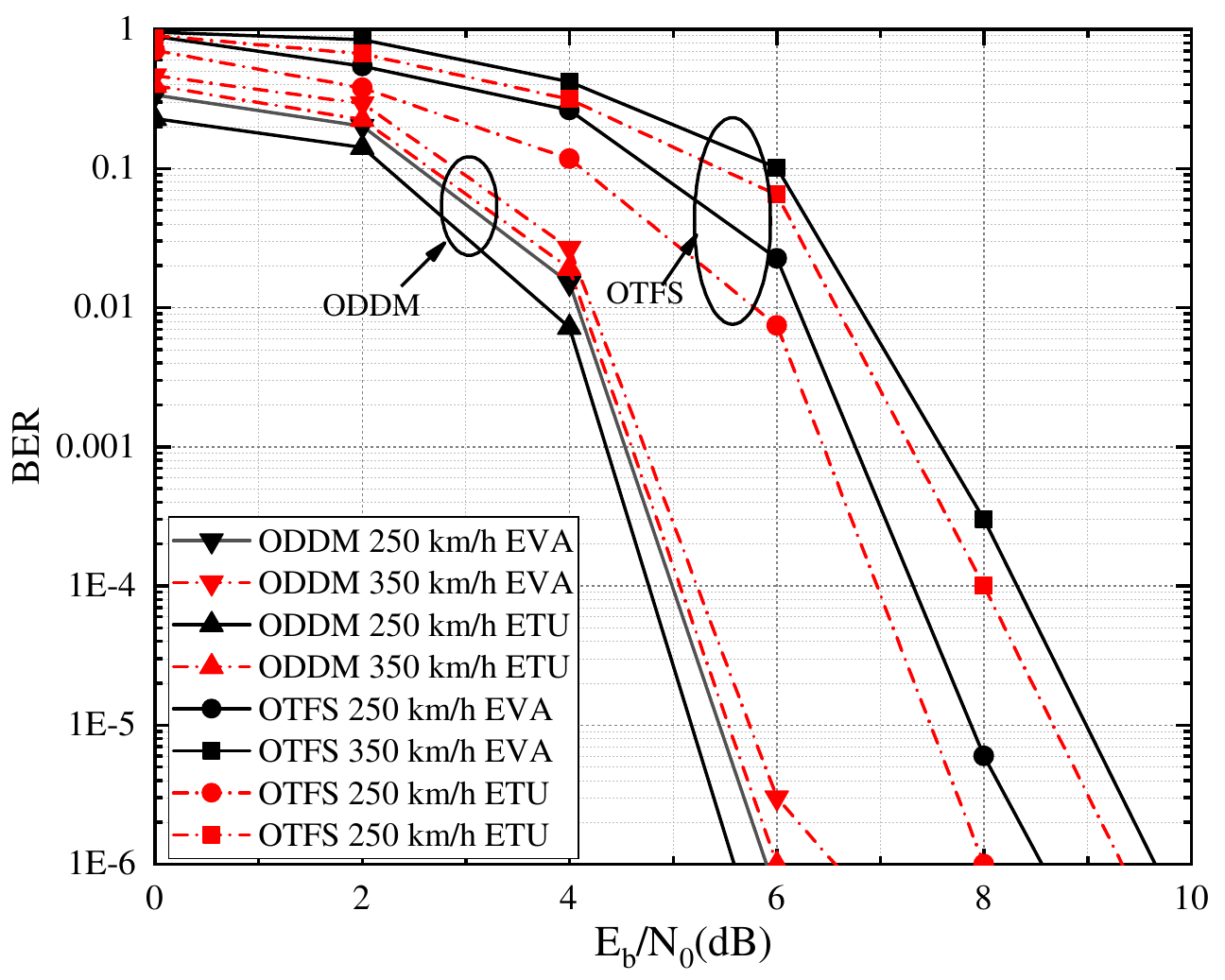}
		\caption{The BER performance versus SNRs for different channel models}
    \label{fig:BER_different_scenarios}
	}
\end{figure}
We also compare the BER performance for EVA and ETU channels in Fig.~\ref{fig:BER_different_scenarios}. It can be seen that compared with the EVA channel, the BER is lower and the symbol detection performance is better under the ETU channel. The reason is that the ETU channel has milder channel characteristic with relatively smaller Doppler spread and moderate multipath delay spread, enabling the ODDM and OTFS systems to achieve better performance.

Finally, in Fig.~\ref{fig:BER_5G_NR}, we compare the performance of ODDM modulation with OFDM modulation in 5G-NR OFDM system.  It can be seen that when the moving speed is low, the performance of ODFM modulation is quite close to that of ODDM modulation. However, in high-speed scenarios, the performance of OFDM modulation is inferior to that of ODDM modulation. This also indirectly reflects that ODDM modulation is capable of combating the performance degradation caused by high-speed movement scenarios. 
\begin{figure}[tbp]
		{	\centering\includegraphics[width=2.7in]{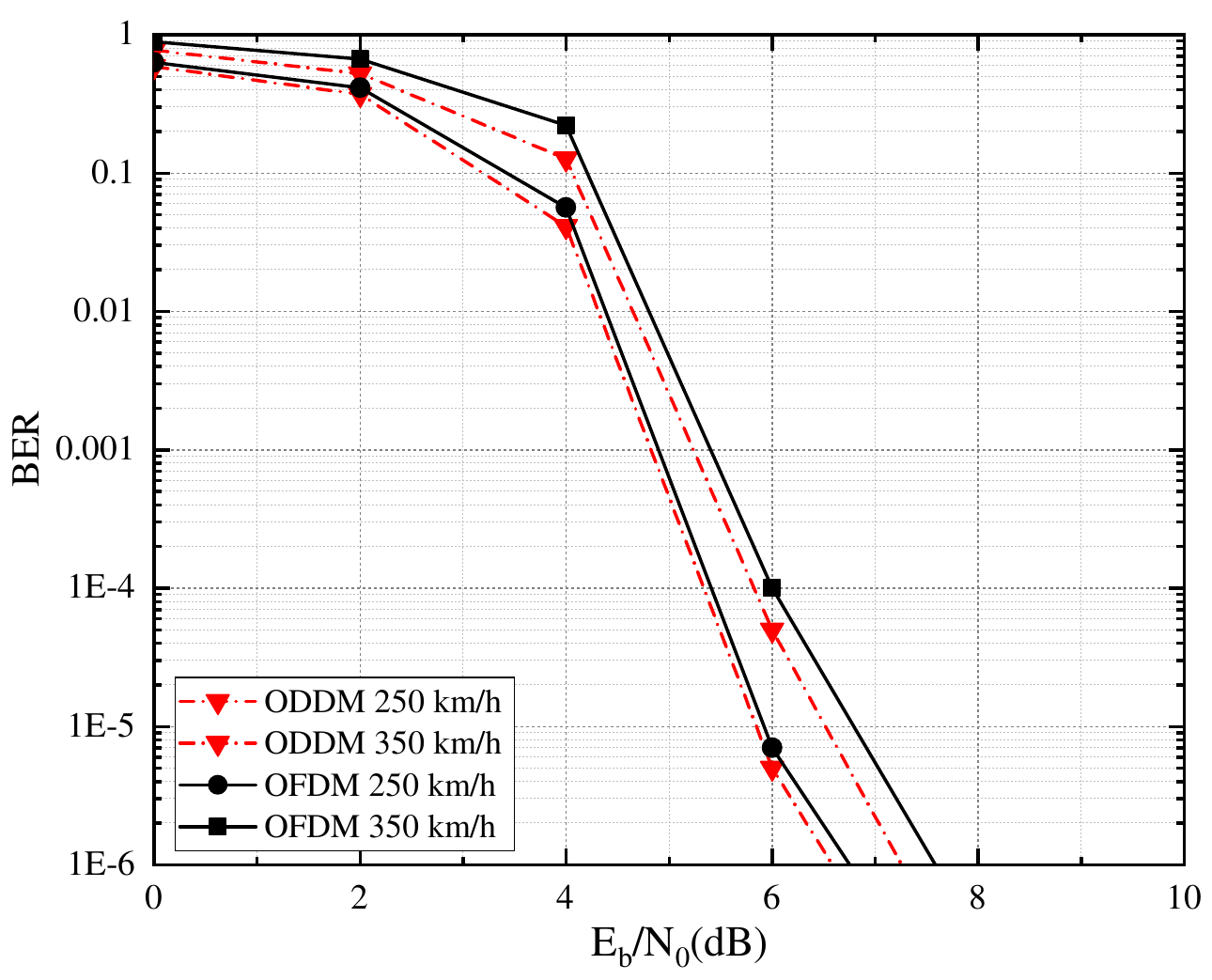}
			\caption{The BER performance versus SNRs in 5G-NR OFDM system}
			\label{fig:BER_5G_NR}
		}
\end{figure}
\section{Conclusion}
In this paper, we delved into the channel estimation problem characterizing massive MIMO-ODDM systems. Since the channel changed rapidly, the channel estimation algorithm should be with low complexity. Therefore, we proposed a two-stage algorithm for the massive MIMO-ODDM systems. Initially, we derived the effective channel model for the system and conducted an in-depth analysis of the characteristics of the effective matrix and channel vector. Based on these characteristics, we first utilized the MAMP technique to precisely identify essential parameters such as channel gain, delay, and Doppler shifts, while the second state employed the DFT method for accurate angle estimation.  The effectiveness of the proposed algorithm was validated through numerical results, which indicates the proposed channel estimation algorithm can approach the Bayesian optimal results as the number of antennas approaches infinity and improve the channel estimation accuracy by about 30\% compared with the existing algorithms in terms of the NMSE metric.

In our future work, we will consider the following possible lead:
			\begin{itemize}
				\item Integrated sensing and communication~(ISAC): ISAC has been seen as a key scenario in the future wireless network. The ODDM modulation has the potential to achieve high reliable communication and high precision sensing. Combining the ODDM modulation with the ISAC system and realizing a balance between sensing and communication performance remains an issue that needs to be addressed.
				\item Scalable multiple access: How to support multiple users in a high-mobility environment is challenging. ODDM modulation provides an opportunity to accommodate the users in the DD domain. Nevertheless, the problem of designing strategies to scale systems for handling multiple users with minimal overhead remains an open topic.
				\item Coded-ODDM: Channel coding is a crucial technique for achieving ultra-reliable communications. The existing literature lacks a dedicated modern channel code design for ODDM systems that can effectively optimize diversity gain using a practical detector and decoder. Therefore, exploring joint detection and decoding strategies for coded-ODDM systems could be a promising approach to fully realize the potential of ODDM.
		\end{itemize}
\appendices

\section{Derivation of Input-Output Relationship}\label{Section_IO_relation}

To derive the input-output relationship, we first extend~\eqref{initial_received_signal} to
 \begin{eqnarray}
&&Y[m,n]=\int e^{-j2\pi n\frac{1}{NT}\Big(t-m\frac{T}{M}\Big)}g_{rx}^*\Big(t-m\frac{T}{M}\Big)y(t)dt\nonumber\\
&&=\int e^{-j2\pi n\frac{1}{NT}\Big(t-m\frac{T}{M}\Big)}g_{rx}^*\Big(t-m\frac{T}{M}\Big)\nonumber\\
&&~~\times\sum_{p=1}^Ph_px (t-\tau_p)e^{j2\pi \nu_p(t-\tau_p)}dt\nonumber\\
&&=\int e^{-j2\pi n\frac{1}{NT}\Big(t-m\frac{T}{M}\Big)}g_{rx}^*\Big(t-m\frac{T}{M}\Big)\nonumber\\
&&~~\times\sum_{p,m',n'}h_pX[m',n']g_{tx}\Big(t-m'\frac{T}{M}-\tau_p\Big)\nonumber\\
&&~~\times e^{j2\pi\Big(\frac{n'}{NT}\Big(t-m'\frac{T}{M}-\tau_p\Big)+\nu_p(t-\tau_p)\Big)}dt\nonumber\\
&&=\int\sum_{p,m',n'}h_p X[m',n']e^{j2\pi\frac{l_p}{NT}(t-k_p\frac{T}{M})}\nonumber\\
&&~~\times g_{rx}^*\Big(t-m\frac{T}{M}\Big)e^{-j2\pi \frac{n}{NT}\Big(t-m\frac{T}{M}\Big)}\nonumber\\
&&~~\times g_{tx}\Big(t-(m'+k_p)\frac{T}{M}\Big)e^{j2\pi\frac{n'}{NT}\Big(t-(m'+k_p)\frac{T}{M}\Big)}dt.\nonumber\\
\end{eqnarray}
where $\tau_p\overset{\triangle}{=}k_p\frac{T}{M}, \nu_p\overset{\triangle}{=}l_p\frac{1}{NT}$. Then, define $t'=t-m\frac{T}{M}$, and based on the orthogonality in~\eqref{orthogonality}, we have
\begin{eqnarray}
&&Y[m,n]=\int\sum_{p,m',n'}h_p X[m',n']e^{j2\pi\frac{l_p}{NT}\big(t'+(m-k_p)\frac{T}{M}\big)}\nonumber\\
&&~~\times g_{rx}^*(t')e^{-j2\pi \frac{n}{NT}t'}g_{tx}\Big(t'-(m'-m+k_p)\frac{T}{M}\Big)\nonumber\\
&&~~\times e^{j2\pi\frac{n'}{NT}\big(t'-(m'-m+k_p)\frac{T}{M}\big)}dt\nonumber\\
&&=\sum_{p,m',n'}h_p X[m',n']e^{j2\pi\frac{l_p}{MN}(m-k_p)}\nonumber\\
&&~~\times\int g_{rx}^*(t')g_{tx}\Big(t'-(m'-m+k_p)\frac{T}{M}\Big)\nonumber\\
&&~~\times e^{-j2\pi\frac{n-l_p-n'}{NT}\big(t'-(m'-m+k_p)\frac{T}{M}\big)} dt\nonumber\\
&&=\sum_{p,m',n'}h_p X[m',n']e^{j2\pi\frac{l_p}{MN}(m-k_p)}\nonumber\\
&&~~\times\delta(m'-m+k_p)\delta(n-l_p-n').
\end{eqnarray}
When $m-k_p<0$, the $(m-k_p)$-th symbol is a $T$ cyclic time-shift of the $(M+m-k_p)$-th symbol. A $T$ cyclic time-shift of an ODDM symbol with subcarrier spacing 
$1/NT$ corresponds to a phase rotation term $e^{\frac{-j2\pi n'NT}{T}}=e^{-j2\pi n'N}$ in the frequency domain, where $\hat{n}=[n-l_p]_N$. Thus, we have

\begin{eqnarray}
&&Y[m,n]=\nonumber\\
&&\left\{
\begin{matrix}
\begin{aligned}
&\sum_ph_pX[m-k_p,\hat{n}]e^{j2\pi\frac{l_p}{MN}(m-k_p)}, \\
&~~~~~~~~~~~~~~~~~~~~~\textrm{if}~m-k_p\geq 0,\\ 
&\sum_ph_pX[M+m-k_p,\hat{n}]e^{j2\pi\left(\frac{l_p}{MN}(m-k_p)-\frac{\hat{n}}{N}\right)}, \\ 
&~~~~~~~~~~~~~~~~~~~~~\textrm{if}~m-k_p<0,
\end{aligned}
\end{matrix}
\right.
\end{eqnarray}
which completes the proof.

\section{Derivation of the Output of the NLE}
From Bayes’ rule, we have 
\begin{align}
    {\rm Pr}(b = 0 \!\mid\! y) &\propto {\rm Pr}(y \!\mid\! b=0) {\rm Pr}(b = 0) \nonumber \\
    &\propto (1-p) {\rm Pr}(y = bg + n_c \!\mid\! b=0) \nonumber \\
    &\propto (1-p) {\rm Pr}(y = n_c) \nonumber \\
    &\propto \frac{1-p}{\sqrt{\delta}} {\rm exp}\Big(-\frac{r^2}{2\delta}\Big), \label{Eqn:5} \\
    {\rm Pr}(b = 1 \!\mid\! y) &\propto {\rm Pr}(y \!\mid\! b=1) {\rm Pr}(b = 1) \nonumber \\
    &\propto p {\rm Pr}(y = bg + n_c \!\mid\! b=1) \nonumber \\
    &\propto p {\rm Pr}(y = g + n_c) \nonumber \\
    &\propto \frac{p}{\sqrt{v_g + \delta}} {\rm exp}\Big(-\frac{(y-u_g)^2}{2(v_g + \delta)}\Big). \label{Eqn:6} 
\end{align}
By normalization, we have 
\begin{eqnarray}
  &&   p_{\rm post}= {\rm Pr}(b=1 \!\mid\! y) \nonumber\\
&&    = \frac{p}{p+c(1-p)\exp(d)},   
\end{eqnarray}
where $c = \sqrt{1+\frac{v_g}{\delta}},\ d = \frac{(y-u_g)^2}{2(\delta+v_g)} - \frac{y^2}{2\delta^2}$. Then, from Bayes’ rule, we have
\begin{align}
    {\rm Pr}(g \!\mid\! y, b=1) &\propto {\rm Pr}(y, b=1 \!\mid\! g){\rm Pr}(g) \\
    &\propto {\rm Pr}(y = g + n_c \!\mid\! g){\rm Pr}(g) \\
    &\propto {\rm exp}\Big(-\frac{(y-g)^2}{2\delta}-\frac{(g-u_g)^2}{2v_g}\Big),
\end{align}
which follows a Gaussian ${\cal N}(\hat{u}_{\rm g}, \hat{v}_{\rm g})$. It is easy to verify $\hat{v}_{\rm g} = (v_g^{-1} + \delta^{-1})^{-1}, \hat{u}_{\rm g} = \hat{v}_{\rm g} (v_g^{-1}u_g + \delta^{-1}y)$. Thus, we have $x_{\rm post} = p_{\rm post}\hat{u}_g, v_{\rm post} = p_{\rm post}(1-p_{\rm post})\hat{u}_g^2 + p_{\rm post}\hat{v}_g$, which completes the proof.

\bibliographystyle{IEEEtran}
\bibliography{sample}
\end{document}